\documentclass{lmcs}
\pdfoutput=1

\usepackage{lastpage}
\lmcsdoi{15}{3}{24}
\lmcsheading{}{\pageref{LastPage}}{}{}%
{Jan.~25,~2018}{Aug.~29,~2019}{}

\usepackage{amsmath}
\usepackage{amssymb}
\usepackage{amsthm}
\usepackage[utf8]{inputenc}
\usepackage{tikz}
\usepackage{float}
\usepackage{indentfirst}
\usepackage{caption}
\usepackage{subcaption}
\usepackage{hyperref}

\makeatletter
\DeclareRobustCommand{\varint}{%
	\mathop{\mathpalette\var@integral\relax}\nolimits
}
\newdimen\var@height
\newcommand{\var@integral}[2]{%
	\sbox\z@{$#1\int$}%
	\var@height=\dimexpr\ht\z@+\dp\z@\relax
	\vcenter{\hbox{%
			\includegraphics[height=\var@height]{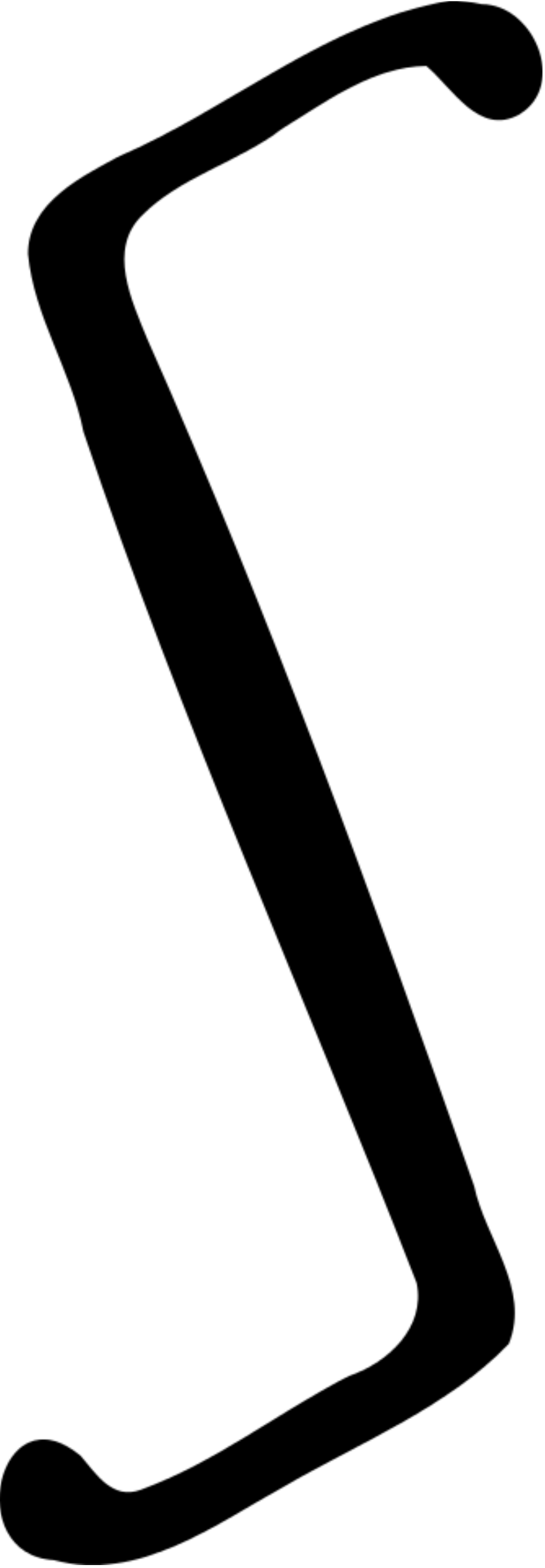}%
		}}%
	}
	\makeatother

\usetikzlibrary{decorations.shapes}
\usetikzlibrary{shapes.geometric}

\newcommand{\ds}{\displaystyle}

\newcommand{\id}{\operatorname{id}}

\newcommand{\ccal}{\mathcal{C}}

\newcommand{\gcal}{\mathcal{G}}
\newcommand{\ical}{\mathcal{I}}
\newcommand{\lcal}{\mathcal{L}}
\newcommand{\mcal}{\mathcal{M}}

\newcommand{\ocal}{\mathcal{O}}

\newcommand{\xcal}{\mathcal{X}}

\newcommand{\nbb}{\mathbb{N}}
\newcommand{\rbb}{\mathbb{R}}

\newcommand{\tbb}{\mathbb{T}}

\newcommand{\psnfamily}{\{\|\cdot\|_n\}_{n\in\nbb}}

\begin{document}

\title{Approximability in the GPAC}
\author[D. Poças]{Diogo Poças\rsuper{a}}
\address{\lsuper{a}Operations Research, Technical University of Munich}
\email{diogo.pocas@tum.de}

\author[J. Zucker]{Jeffery Zucker\rsuper{b}}
\address{\lsuper{b}Department of Computing and Software, McMaster University}


\begin{abstract}

Most of the physical processes arising in nature are modeled by differential equations, either ordinary (example: the spring/mass/damper system) or partial (example: heat diffusion). From the point of view of analog computability, the existence of an effective way to obtain solutions (either exact or approximate) of these systems is essential.

A pioneering model of analog computation is the General Purpose Analog Computer (GPAC), introduced by Shannon as a model of the Differential Analyzer and improved by Pour-El, Lipshitz and Rubel, Costa and Graça and others. The GPAC is capable of manipulating real-valued data streams. Its power is known to be characterized by the class of differentially algebraic functions, which includes the solutions of initial value problems for ordinary differential equations.

We address one of the limitations of this model, concerning the notion of approximability, a desirable property in computation over continuous spaces that is however absent in the GPAC. In particular, the Shannon GPAC cannot be used to generate non-differentially algebraic functions such as the gamma function, which can be approximately computed in other models of computation. We extend the class of data types using networks with channels which carry information on a general complete metric space $\xcal$; for example $\xcal=C(\rbb,\rbb)$, the class of continuous functions of one real (spatial) variable.

We consider the original modules in Shannon's construction (constants, adders, multipliers, integrators) and we add \emph{(continuous or discrete) limit} modules which have one input and one output. For input $u$, they output the continuous limit $\ds g=\lim_{t\rightarrow\infty}u(t)$ or the discrete limit $\ds g=\lim_{n\rightarrow\infty}u_n$.

We then define an L-GPAC to be a network built with $\xcal$-stream channels and the above-mentioned modules. This leads us to a framework in which the specifications of such analog systems are given by fixed points of certain operators on continuous data streams. Such a framework was considered by Tucker and Zucker. We study these analog systems and their associated operators, and show how some classically non-generable functions, such as the gamma function and the Riemann zeta function, can be captured with the L-GPAC.
\end{abstract}

\maketitle

\section{Introduction}

Analog computation, as conceived by Kelvin \cite{thompsontait:80}, Bush \cite{bush:31}, and Hartree \cite{hartree:50}, is a form of experimental computation with physical systems called analog devices or analog computers. Historically, data are represented by measurable physical quantities, including lengths, shaft rotation, voltage, current, resistance, etc., and the analog devices that process these representations are made from mechanical, electromechanical or electronic components \cite{small:93,holst:96,johansson:96}.

A general purpose analog computer (GPAC) was introduced by Shannon \cite{shannon:41} as a model of Bush’s Differential Analyzer \cite{bush:31}. Shannon discovered that a function can be generated by a GPAC if, and only if, it is differentially algebraic, but his proof was incomplete. A basic study was made by Pour-El \cite{pourel:74} who gave some good characterizations of the class of analog computable functions, focusing on the classical analog systems built from constants, adders, multipliers and integrators. This yielded a stronger model and a new proof of the Shannon’s equivalence (and some new gaps, corrected by Lipshitz, Rubel \cite{lipshitzrubel:87}, Graça and Costa \cite{gracacosta:03}). Using this characterization in terms of algebraic differential equations, Pour-El showed that not all computable functions on the reals (in the sense of computable analysis) can be obtained with these analog networks. The well-known counterexample is the gamma function
\[\Gamma(t)=\int_0^\infty x^{t-1}e^{-x}dx\]
which is not differentiable algebraic and so cannot be generated by a GPAC, as noted by Shannon himself. However, one could expect that, in a ``sensible'' model of computability on continuous data, this function would be computable.

Indeed, the gamma function is \emph{effectively computable} in the sense of computable analysis, a branch of constructive mathematics studied by Grzegorczyk \cite{grzegorczyk:55,grzegorczyk:57}, Lacombe \cite{lacombe:55a,lacombe:55b,lacombe:55c}, Pour-El, Richards \cite{pourelrichards:79}, Weihrauch \cite{weihrauch:00}, Tucker, Zucker \cite{tuckerzucker:07}, among others. These researchers have in one way or another tried to answer what is perhaps the fundamental question for analog computation: which functions are computable? For the case of digital computation, all empirical evidence corroborates the celebrated Church-Turing Thesis, showing that various models (such as Turing machines, $\lambda$-calculus and recursive functions) are equivalent. This picture is not so clear for analog computability on continuous spaces, despite the abundance of progress made and (partial) equivalence results \cite{olivieretal:06,ko:91,hansentucker:99,weihrauch:00}.

Returning to the Shannon GPAC and the (non-)computability of the gamma function, some researchers have attempted to include \emph{approximability} in the model (which is an important ingredient in many models of real computation such as computable analysis); in particular, Graça \cite{graca:04} redefined the notion of GPAC-computability in order to show that the gamma function can indeed be considered as GPAC-computable.

We are interested in defining models of computation via analog networks that extend well-known results into spaces of functions of several variables. In our previous paper \cite{pocaszucker:16} we considered a modest family of linear evolution problems on Fréchet spaces. We studied how to describe solutions to these problems as fixed points of certain analog networks and, inspired by Cauchy-Kowalevski theorems, we attempted to produce such fixed points via iterating sequences. In a later paper \cite{pocaszucker:18} we presented an extension to the Shannon GPAC model, which we called $X$-GPAC, that can reason about functions of more than one variable. In that paper we considered a \emph{differential module} (that produces spatial derivatives) and we established an equivalence theorem, characterizing the class of $X$-GPAC-generable functions in terms of solutions to \emph{partial differential algebraic systems of equations}.

In this paper we present a different extension to the Shannon GPAC model. In particular, we wish to incorporate the procedure of taking limits into our model of analog networks. In abstract terms, one may want to define a class of `computable' elements $\ccal$ such that

\begin{center}If $f\in\ccal$, then $\lim f\in\ccal$.\end{center}

Of course, part of the problem is understanding what kinds of `limit' we are allowed to consider. Usually, in computability theory on continuous spaces, we must demand that limits be `effective', in the sense that the modulus of convergence is known \text{a priori} and thus we can effectively obtain an approximation to the limit within a prescribed precision. The notion of limit must also agree with the topology of the underlying space, which can be induced by a metric, a norm, or a family of pseudonorms. Thus if $\xcal$ is a function space we may be interested in `uniform' or `locally uniform' as opposed to `pointwise' limits.

The paper is structured as follows. We begin by introducing the Shannon GPAC and how we intend to extend it with additional channel types. We then define the notions of Cauchy sequence, Cauchy stream and effective convergence. With those ingredients, we are able to consider a new module that takes (discrete or continuous) limits and an extension to the Shannon GPAC, which we call L-GPAC. Afterwards we prove some auxiliary results and briefly discuss other interesting (yet equivalent) possibilities for limits and limit modules. In the latter part of the paper, we show how to generate some non-differentially algebraic functions, such as the gamma and Riemann zeta functions, which are our main motivation for including limits.

We briefly summarize the original content of this paper. It is important to remark that the idea of introducing approximability into the GPAC model is not new and can attributed to Graça, \cite{graca:04}. In particular, the paper \cite{bournezetal:07} provides a notion of GPAC-computability also based on limits, remarkably showing an equivalence with the class of computable functions on a compact interval. However, we claim that (to the best of our knowledge) the approach of including a \emph{module} that performs limits is original. In our framework, approximability is incorporated on the GPAC model itself, and not just in the way we define the GPAC semantics. This is how we suggest our results be contrasted to those of Bournez, Campagnolo, Graça, Hainry and other authors. Therefore, the main original content of this paper consists in the introduction of limit modules (Definition \ref{def:limitmodules}) and the notion of L-GPAC (Definition \ref{def:lgpac}); the computability of the gamma and Riemann zeta functions (Theorems \ref{thm:compgammafun} and \ref{thm:compzetafun}) can be seen as applications of this theory.

\section{Channels and modules}\label{sec:disctype}

The main objects of our study are \emph{analog networks} or \emph{analog systems}, \cite{tuckerzucker:07, tuckerzucker:11, jameszucker:13, tuckerzucker:14}, whose main components are described as follows:

\begin{center}
\emph{Analog network} = \emph{data} + \emph{time} + \emph{channels} + \emph{modules}.
\end{center}

We can model \emph{data} as elements of a complete metric vector space $\xcal$, such as a Banach or Fréchet space. We will use the nonnegative real numbers as a continuous model of \emph{time} $\tbb=[0,\infty)$. There are two types of \emph{channels} we can consider: a \emph{scalar channel} carries a constant value $x\in\xcal$, whereas a \emph{stream channel} carries a continuously differentiable stream, represented as a function $u:\tbb\rightarrow \xcal$ (this space is denoted by $C^1(\tbb,\xcal)$). Each \emph{module} $M$ has zero, one or more input channels, and must have a single output channel; thus it can be specified by a (possibly partially defined) \emph{stream function}
\[F_M:\xcal^k\times C^1(\tbb,\xcal)^\ell\rightharpoonup C^1(\tbb,\xcal)\quad (k,\ell\geq 0).\]
In the case that $\xcal=\rbb$, we obtain the Shannon GPAC. We can use four types of modules to describe this model, which are equivalent to the ones originally used by Shannon.

\begin{defi}[\textbf{Shannon modules}]\label{def:smods} The \emph{Shannon modules} are defined as follows:

\begin{itemize}
\item for each $c\in\rbb$, there is a \emph{constant module} with zero inputs and one output $v$, given by
\[v(t)=c;\]
\item the \emph{adder module} has two inputs $u$, $v$ and one output $w$, given by
\[w(t)=u(t)+v(t);\]
\item the \emph{multiplier module} has two inputs $u$, $v$  and one output $w$, given by
\[w(t)=u(t)v(t);\]
\item the \emph{integrator module} has a scalar input $c$ (also called initial setting), two stream inputs $u$, $v$ and one output $w$, given by the Riemann-Stieltjes integral
\[w(t)=c+\int_0^t u(s)v'(s)ds;\]
\end{itemize}
\end{defi}
We can depict the four Shannon modules in box diagrams, as in Figure \ref{fig:shannonmodules}. We also introduce the symbol `$\varint$' to denote the operator associated with the integrator module, in order to differentiate from the actual integral; we can then write $\varint(c,u,v)=c+\int udv$.

\begin{figure}[ht]
\flushleft
\hspace{1.2cm}
\vspace{.5cm}
\begin{tikzpicture}
\draw (-.5, .5) -- ( .5, .5) -- ( .5,-.5) -- (-.5,-.5) -- (-.5, .5);
\draw[->] (  .5, 0) -- (1.5, 0);

\node at ( 0,0  ) {$c$}; 
\node at ( 1,.25) {$c$};
\node[anchor=west] at (2.5, .5) {$c:\rightarrow C^1(\tbb,\rbb)$};
\node[anchor=west] at (2.5,-.5) {$c(t)=c$};
\end{tikzpicture}\\
\hspace{0.3cm}
\vspace{.5cm}
\begin{tikzpicture}
\draw (-.5, .5) -- ( .5, .5) -- ( .5,-.5) -- (-.5,-.5) -- (-.5, .5);
\draw[->] (-1.5, .25) -- (-.5, .25);
\draw[->] (-1.5,-.25) -- (-.5,-.25);
\draw[->] (.5, 0) -- (1.5, 0);

\node at ( 0, 0 ) {$+$};
\node at (-1, .5) {$u$};
\node at (-1, 0 ) {$v$};
\node at ( 1,.25) {$u+v$};
\node[anchor=west] at (2.5, .5) {$+:C^1(\tbb,\rbb)\times C^1(\tbb,\rbb)\rightarrow C^1(\tbb,\rbb)$};
\node[anchor=west] at (2.5,-.5) {$+(u,v)(t)=u(t)+v(t)$};
\end{tikzpicture}\\
\hspace{0.3cm}
\vspace{.5cm}
\begin{tikzpicture}
\draw (-.5, .5) -- ( .5, .5) -- ( .5,-.5) -- (-.5,-.5) -- (-.5, .5);
\draw[->] (-1.5, .25) -- (-.5, .25);
\draw[->] (-1.5,-.25) -- (-.5,-.25);
\draw[->] (.5, 0) -- (1.5, 0);

\node at ( 0, 0 ) {$\times$};
\node at (-1, .5) {$u$};
\node at (-1, 0 ) {$v$};
\node at ( 1,.25) {$uv$};
\node[anchor=west] at (2.5, .5) {$\times:C^1(\tbb,\rbb)\times C^1(\tbb,\rbb)\rightarrow C^1(\tbb,\rbb)$};
\node[anchor=west] at (2.5,-.5) {$\times(u,v)(t)=u(t)v(t)$};
\end{tikzpicture}\\
\hspace{0.3cm}
\vspace{.5cm}
\begin{tikzpicture}
\draw (-.5, .5) -- ( .5, .5) -- ( .5,-.5) -- (-.5,-.5) -- (-.5, .5);
\draw[->] (-1.5, .375) -- (-.5, .375);
\draw[->] (-1.5, 0  ) -- (-.5, 0  );
\draw[->] (-1.5,-.375) -- (-.5,-.375);
\draw[->] (  .5, 0  ) -- (1.5, 0  );

\node at ( 0 , 0 ) {$\varint$};
\node at (-1 , .5  ) {$c$};
\node at (-1 , .125) {$u$};
\node at (-1 ,-.25 ) {$v$};
\node at (1.3,.25) {$c+\int udv$};
\node[anchor=west] at (2.5, .5) {$\varint: \rbb\times C^1(\tbb,\rbb)\times C^1(\tbb,\rbb)\rightarrow C^1(\tbb,\rbb)$};
\node[anchor=west] at (2.5,-.5) {$\varint(c,u,v)(t)=c+\int_0^t u(s)dv(s)$};
\end{tikzpicture}\caption{The four Shannon modules.\label{fig:shannonmodules}}
\end{figure}
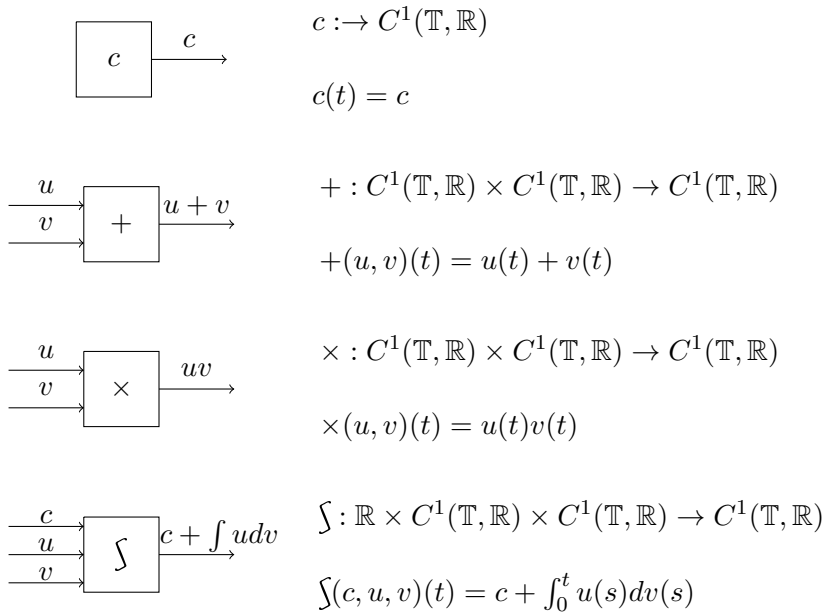

The continuous differentiability of the streams $u,v$ ensures that the integrator module is well defined; indeed, the Riemann-Stieltjes integral is well defined for continuous integrand and continuously differentiable integrator. In other words, for any $u,v\in C^1(\tbb,\rbb)$, the formula $\int_0^t u(s)dv(s)$ defines a function in $C^1(\tbb,\rbb)$.

In a previous paper \cite{pocaszucker:18} we presented an extension of the Shannon GPAC, which we called $\xcal$-GPAC, that allowed the study of functions of more than one variable. The main idea present in that paper is to extend the \emph{output space}, that is, replacing $C^1(\tbb,\rbb)$ with $C^1(\tbb,\xcal)$, where $\xcal$ is a metric vector space. For example, we can think of $\xcal$ as the space of continuous real-valued functions on $\rbb^n$, that is, $\xcal=C(\rbb^n,\rbb)$. In this way, our channels will now carry $\xcal$-valued streams of data $u:\tbb\rightarrow \xcal$, which correspond to functions of $n+1$ real variables, under the ``uncurrying''
\[\tbb\rightarrow(\rbb^n\rightarrow\rbb)\simeq\tbb\times\rbb^n\rightarrow\rbb.\]

It is evident that one of the independent variables, namely the ``time'' variable, plays a different role from the others. Our approach is, to some extent, motivated by the theory of partial differential equations, in which some fundamental problems (such as the heat equation, wave equation and Schrödinger equation) can be expressed as time evolution problems in a function space.

In this paper we will present another extension of the Shannon GPAC, which can be described as a \emph{multityped GPAC}. We consider a metric vector space of the form $\xcal=C(\Omega,\rbb)$, where $\Omega$ is a closed interval in $\rbb$ (either bounded or unbounded). This model comes equipped with four channel types:

\begin{itemize}
\item \emph{$\rbb$-scalar} channels, which carry a constant $k\in\rbb$;
\item \emph{$\xcal$-scalar} channels, which carry a constant $x\in \xcal$;
\item \emph{$\rbb$-stream} channels, which carry a stream $a\in C^1(\tbb,\rbb)$;
\item \emph{$\xcal$-stream} channels, which carry a stream $u\in C^1(\tbb,\xcal)$.
\end{itemize}
The reason for introducing a variety of channel types will become clear when we introduce the limit module (Definition \ref{def:limitmodules}). We also observe that the basic Shannon modules (constants, adders, multipliers and integrators) can be easily generalized from $\rbb$-channels into $\xcal$-channels.

\begin{rem}[\textbf{Discrete channel types}] We mention in passing that \emph{discrete} channel types could also be considered. If we wished to do so, the resulting model could be seen as a hybrid between discrete and continuous computation. The addition of more channels would undoubtedly increase the difficulty of studying the power of the GPAC; we make the important remark that these discrete channel types are not essential to the main purpose of this paper, which is to generate some non-differentially algebraic functions. Therefore they are included only as an illustration. In any case, here are the further channel types we may wish to consider:

\begin{itemize}
\item \emph{$\nbb$-scalar} channels, which carry a constant $k\in\nbb$;
\item \emph{$\nbb$-sequence} channels, which carry a sequence $\{k_n\}\in\nbb^\nbb$;
\item \emph{$\rbb$-sequence} channels, which carry a sequence $\{k_n\}\in\rbb^\nbb$;
\item \emph{$\xcal$-sequence} channels, which carry a sequence $\{g_n\}\in \xcal^\nbb$.
\end{itemize}
We remark that the channel type corresponding to $\nbb$-streams is not necessary, since any continuous function of type $\tbb\rightarrow\nbb$ must be constant.
\end{rem}

\section{The limit operator and the limit GPAC}

Let us make precise what we mean by effective limit. If we take $\xcal$ to be a complete metric space with a metric $d$, then

\begin{itemize}
\item a sequence $\{g_n\}\in \xcal^\nbb$ is a \emph{Cauchy sequence} whenever
\[\text{for all }\epsilon>0\text{ there exists }N\in\nbb\text{ such that for }m,n\in\nbb\text{ with }m,n\geq N\text{ one has }d(g_m,g_n)<\epsilon;\]
\item a stream $u\in C^1(\tbb,\xcal)$ is a \emph{Cauchy stream} whenever
\[\text{for all }\epsilon>0\text{ there exists }T\in\tbb\text{ such that for }s,t\in\tbb\text{ with }s,t\geq T\text{ one has }d(u(s),u(t))<\epsilon;\]
\end{itemize}
To write the effective version of these limits, we begin by replacing the existential quantifiers with functions on the precision $\epsilon$. A possible approach is given in the following definitions.

\begin{defi}[\textbf{Moduli of convergence}]\label{def:modconv}~
\begin{enumerate}
\item A \emph{discrete modulus of convergence} is a nondecreasing function $N:\nbb\rightarrow\nbb$.
\item A \emph{continuous modulus of convergence} is a nondecreasing function $T\in C^1(\tbb,\rbb)$.
\end{enumerate}
\end{defi} 

\begin{rem}[\textbf{Effective moduli of convergence}]\label{rmk:effmodconv}
Both definitions of moduli of convergence can be effectivized in an intuitive manner. To effectivize the notion of discrete modulus of convergence $N:\nbb\rightarrow\nbb$, we can require that $N$ be computable (in the traditional sense). To effectivize the notion of continuous modulus of convergence $T\in C^1(\tbb,\rbb)$, we can require that $T$ be GPAC-generable. In the latter case we may further specify what type of GPAC we are interested: either the Shannon GPAC or the limit GPAC which we will develop in this paper. However, for most of the time we will desire $T$ to be a somewhat ``simple'' function, such as a monomial, an exponential, or a chain of exponentials, in which case the notion of Shannon GPAC-generability suffices. In fact, we expect the computational richness of the construction to be in the stream for which we are taking limits, but not on the modulus of convergence itself.
\end{rem}

\begin{defi}[\textbf{Effective limits on metric spaces}]\label{def:efflimmetric}~
\begin{enumerate}
\item Let $N$ be a discrete modulus of convergence and $\{g_n\}\in \xcal^\nbb$. Then $\{g_n\}$ is an \emph{$N$-convergent Cauchy sequence} if
\[\text{for all }\nu\in\nbb\text{, for all }m,n\in\nbb\text{ with }m,n\geq N(\nu)\text{ one has }d(g_m,g_n)<2^{-\nu}.\]
\item Let $T$ be a continuous modulus of convergence and $u\in C^1(\tbb,\xcal)$. Then $u$ is a \emph{$T$-convergent Cauchy stream} if
\[\text{for all }\tau\in\tbb\text{, for all }s,t\in\tbb\text{ with }s,t\geq T(\tau)\text{ one has }d(u(s),u(t))<2^{-\tau}.\]
\item A sequence $\{g_n\}\in \xcal^\nbb$ is called an \emph{effective Cauchy sequence} if there is an effective discrete modulus of convergence $N$ such that $\{g_n\}$ is $N$-convergent.
\item A stream $u\in C^1(\tbb,\xcal)$ is called an \emph{effective Cauchy stream} if there is an effective continuous modulus of convergence $T$ such that $u$ is $T$-convergent.
\end{enumerate}
\end{defi}
\noindent An example of a modulus of convergence is given by the identity function, either discrete ($\id:\nbb\rightarrow\nbb$) or continuous ($\id\in C^1(\tbb,\rbb)$). We note that any effective Cauchy sequence may be effectively replaced by an $\id$-convergent Cauchy sequence via a composition with its modulus of convergence; in other words, if $\{g_n\}$ is an $N$-convergent Cauchy sequence, then $\{g_{N(n)}\}$ is an $\id$-convergent Cauchy sequence. Similarly, an effective Cauchy stream may be effectively replaced by an $\id$-convergent Cauchy stream. Thus we may assume, for convenience, that the modulus of convergence for a given effective limit is given by the identity map.

The final ingredient of our construction is the notion of a \emph{limit operator}, and again this may be done in a discrete or continuous manner.

\begin{defi}[\textbf{Limit modules}]\label{def:limitmodules}~
\begin{enumerate}
\item For the data type $\xcal$, there is a \emph{discrete limit module} with one input of type $\xcal^\nbb$ and one output of type $\xcal$. For input $\{g_n\}$, it outputs the $\id$-convergent limit $\ds\lim_{n\rightarrow\infty}g_n$ (if it exists).
\item For the data type $\xcal$, there is a \emph{continuous limit module} with one input of type $C^1(\tbb,\xcal)$ and one output of type $\xcal$. For input $u$, it outputs the $\id$-convergent limit $\ds\lim_{t\rightarrow\infty}u(t)$ (if it exists).
\end{enumerate}
\end{defi}

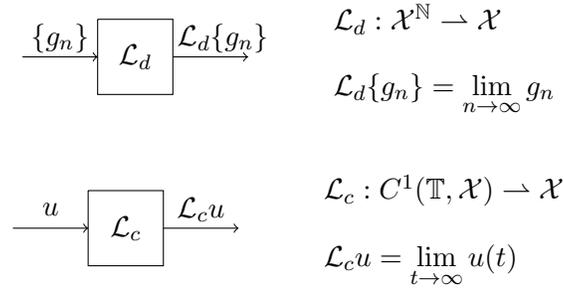
\begin{figure}[ht]\centering
\begin{tikzpicture}
\draw (-.5, .5) -- ( .5, .5) -- ( .5,-.5) -- (-.5,-.5) -- (-.5, .5);
\draw[->] (-1.5,0) -- (-.5,0);
\draw[->] (.5, 0) -- (1.5, 0);

\node at ( 0 , 0) {$\lcal_d$};
\node at (-1,.25) {$\{g_n\}$};
\node at ( 1.1,.25) {$~\lcal_d \{g_n\}$};
\node[anchor=west] at (2.5, .5) {$\lcal_d:\xcal^\nbb\rightharpoonup \xcal$};
\node[anchor=west] at (2.5,-.5) {$\ds\lcal_d \{g_n\}=\lim_{n\rightarrow\infty}g_n$};
\end{tikzpicture}\\
\vspace{.5cm}
\begin{tikzpicture}
\draw (-.5, .5) -- ( .5, .5) -- ( .5,-.5) -- (-.5,-.5) -- (-.5, .5);
\draw[->] (-1.5,0) -- (-.5,0);
\draw[->] (.5, 0) -- (1.5, 0);

\node at ( 0 , 0) {$\lcal_c$};
\node at (-1,.25) {$u$};
\node at ( 1,.25) {$\lcal_c u$};
\node[anchor=west] at (2.5, .5) {$\lcal_c:C^1(\tbb,\xcal)\rightharpoonup \xcal$};
\node[anchor=west] at (2.5,-.5) {$\ds\lcal_c u=\lim_{t\rightarrow\infty}u(t)$};
\end{tikzpicture}\caption{Limit modules.\label{fig:limitmodules}}
\end{figure}

A few comments are in order. Firstly, it should be clear that the limit modules define partial-valued operators; they are only defined for those sequences in $\xcal^\nbb$ (or those functions in $C^1(\tbb,\xcal)$) that have an $\id$-convergent limit. Secondly, the choice of the identity as the `canonical' modulus of convergence allows us to specify the limit operator as a one-input, one-output module. A different approach could be taken, in which a \emph{two-input limit module} is considered, having one input for the sequence (or stream) and another input for the discrete (or continuous) modulus of convergence, such as in Figure \ref{fig:limittwoinputs}.

\begin{figure}[ht]\centering
\begin{tikzpicture}
\draw (-.5, .5) -- ( .5, .5) -- ( .5,-.5) -- (-.5,-.5) -- (-.5, .5);
\draw[->] (-1.5, .25) -- (-.5, .25);
\draw[->] (-1.5,-.25) -- (-.5,-.25);
\draw[->] (.5, 0) -- (1.5, 0);

\node at ( 0 , 0) {$\tilde{\lcal}_d$};
\node at (-1,.5) {$\{g_n\}$};
\node at (-1, 0 ) {$N$};
\node at ( 1.1,.25) {$~\tilde{\lcal}_d \{g_n\}$};
\node[anchor=west] at (2.5, .5) {$\tilde{\lcal}_d:\xcal^\nbb\times\nbb^\nbb\rightharpoonup \xcal$};
\node[anchor=west] at (2.5,-.5) {$\ds\tilde{\lcal}_d (\{g_n\},N)=\lim_{n\rightarrow\infty}g_n$};
\end{tikzpicture}\\
\vspace{.5cm}
\begin{tikzpicture}
\draw (-.5, .5) -- ( .5, .5) -- ( .5,-.5) -- (-.5,-.5) -- (-.5, .5);
\draw[->] (-1.5, .25) -- (-.5, .25);
\draw[->] (-1.5,-.25) -- (-.5,-.25);
\draw[->] (.5, 0) -- (1.5, 0);

\node at ( 0 , 0) {$\tilde{\lcal}_c$};
\node at (-1,.5) {$u$};
\node at (-1, 0 ) {$T$};
\node at ( 1,.25) {$\tilde{\lcal}_c u$};
\node[anchor=west] at (2.5, .5) {$\tilde{\lcal}_c:C^1(\tbb,\xcal)\times C^1(\tbb,\rbb)\rightharpoonup \xcal$};
\node[anchor=west] at (2.5,-.5) {$\ds\tilde{\lcal}_c (u,T)=\lim_{t\rightarrow\infty}u(t)$};
\end{tikzpicture}\caption{Two-input limit modules.\label{fig:limittwoinputs}}
\end{figure}
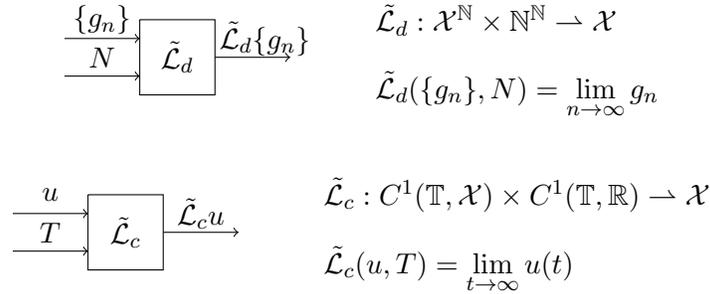

Since, as explained above, any effective limit may be converted to an $\id$-convergent limit via a composition with the modulus of convergence, we can derive the two-input limit module from the one-input limit module using a composition, as depicted on Figure \ref{fig:limitcomposition}.

\begin{figure}[ht]\centering
\raisebox{-.5\height}{\begin{tikzpicture}
\draw (-.5, 1.25) -- ( .5, 1.25) -- ( .5,  .25) -- (-.5,  .25) -- (-.5, 1.25);
\draw (-.5,- .25) -- ( .5,- .25) -- ( .5,-1.25) -- (-.5,-1.25) -- (-.5,- .25);
\draw (1.5,  .5 ) -- (2.5,  .5 ) -- (2.5,- .5 ) -- (1.5,- .5 ) -- (1.5,  .5 );
\draw[->] (-1.5, 0  ) -- (-1  , 0  ) -- (-1  , .75) -- (-.5, .75);
\draw[->] (-1  , 0  ) -- (-1  ,-.75) -- (-.5 ,-.75);
\draw[->] (  .5, .75) -- ( 1  , .75) -- ( 1  , .25) -- (1.5, .25);
\draw[->] (  .5,-.75) -- ( 1  ,-.75) -- ( 1  ,-.25) -- (1.5,-.25);
\draw[->] ( 2.5, 0  ) -- ( 4.5, 0  );

\node at ( 0   , .75) {$u$};
\node at ( 0   ,-.75) {$T$};
\node at ( 2   , 0  ) {$\tilde{\lcal}_c$};
\node at (-1.25, .25) {$t$};
\node at (  .75, 1  ) {$u$};
\node at (  .75,-.5 ) {$T$};
\node at ( 3.5 , .25) {$\tilde{\lcal}_c(u,T)$};
\end{tikzpicture}}
\hspace{.5cm}
\raisebox{-.5\height}{\begin{tikzpicture}
\draw (-.5, .5) -- ( .5, .5) -- ( .5,-.5) -- (-.5,-.5) -- (-.5, .5);
\draw (1.5, .5) -- (2.5, .5) -- (2.5,-.5) -- (1.5,-.5) -- (1.5, .5);
\draw (3.5, .5) -- (4.5, .5) -- (4.5,-.5) -- (3.5,-.5) -- (3.5, .5);
\draw[->] (-1.5,0) -- (-.5,0);
\draw[->] (  .5,0) -- (1.5,0);
\draw[->] ( 2.5,0) -- (3.5,0);
\draw[->] ( 4.5,0) -- (5.5,0);

\node at ( 0, 0 ) {$T$};
\node at ( 2, 0 ) {$u$};
\node at ( 4, 0 ) {$\lcal_c$};
\node at (-1,.25) {$t$};
\node at ( 1,.25) {$T$};
\node at ( 3,.25) {$u\circ T$};
\node at (5.3,.25) {$~\lcal_c(u\circ T)$};
\end{tikzpicture}}
\caption[Derivation of the two-input continuous limit module.]{Derivation of the two-input continuous limit module; the discrete case is done similarly.\label{fig:limitcomposition}}
\end{figure}
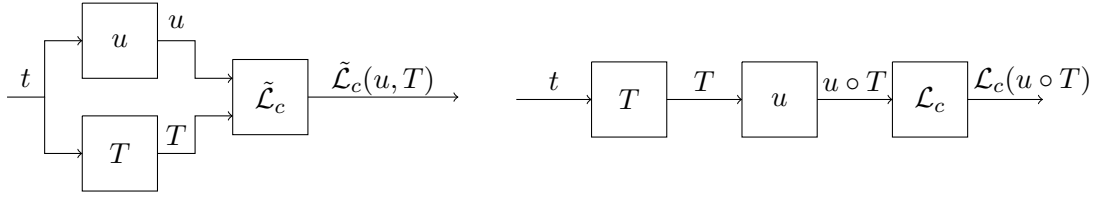

\begin{defi}[\textbf{L-GPAC}]\label{def:lgpac}
A \emph{limit general purpose analog computer} (L-GPAC) is a network built with:
\begin{itemize}
\item $\rbb$-channels and $\xcal$-channels (carrying either constants or streams);
\item the basic modules (constants, adders, multipliers, integrators) and the one-input continuous limit module.
\end{itemize}
Moreover, the channels connect the inputs and outputs of the modules, with the following restrictions:
\begin{itemize}
\item the only connections allowed are between an output and an input;
\item each input may be connected to either zero or one output;
\end{itemize}
\end{defi}

\begin{rem}[\textbf{L-GPAC semantics}]\label{rmk:lgpac}
We must mention that some non-obvious choices were made in Definition \ref{def:lgpac}:

\begin{itemize}
\item Should we have included the discrete channel types from Section \ref{sec:disctype} and the discrete limit module from Definition \ref{def:limitmodules}?
\item Should we have opted for the two-input instead of one-input limit modules?
\item Should we have considered various notions of effective moduli of convergence (cf. Remark~\ref{rmk:effmodconv})?
\end{itemize}
Clearly, as we increase the variety (in both channel types and modules) of our construction, we get more inclusive models of computation, but finding characterization results becomes increasingly difficult and technical. Keeping in mind that our goal is to compute some non-differentially algebraic functions such as the gamma function and the Riemann zeta function, we can limit our construction to the minimum that makes that goal achievable. As will be seen in Sections \ref{sec:gammafun} and \ref{sec:compzetafun}, this can be accomplished by considering only continuous channel types and a one-input continuous limit module.
\end{rem}

We can follow the same program as in \cite{pocaszucker:18} and define the notions of induced operator, well-posedness, L-GPAC semantics and L-GPAC-generability which follow from the notion of L-GPAC. We shall only sketch this construction; the interested reader may wish to consult \cite{pocaszucker:18} for more details.

\begin{defi}[\textbf{L-GPAC semantics}]\label{def:lgpacsemantics}
Given an L-GPAC $\gcal$,
\begin{enumerate}
\item we denote the \emph{induced operator} (also called \emph{input-output operator}) by $\Phi:\ical\times\mcal\rightharpoonup\mcal\times\ocal$, where $\ical$ corresponds to \emph{proper input} channels, $\ocal$ corresponds to \emph{proper output} channels and $\mcal$ corresponds to \emph{mixed input/output} channels;
\item the \emph{fixed point equation} is given by
\begin{equation}\label{eq:lgpacfp} \Phi(\textbf{inp},\textbf{mix})=(\textbf{mix},\textbf{out}),\end{equation}
where $\textbf{inp}\in\ical$, $\textbf{mix}\in\mcal$, $\textbf{out}\in\ocal$;
\item consider an open subset $U\subseteq\ical$; we say that $\gcal$ is \emph{well-posed} on $U$ if for all $\textbf{inp}\in U$ there is a unique $(\textbf{mix},\textbf{out})$ such that \eqref{eq:lgpacfp} holds; and moreover the map $\textbf{inp}\mapsto(\textbf{mix},\textbf{out})$ is continuous;
\item if $\gcal$ is well-posed on $U$, then we say it \emph{generates} the function $F:\ical\rightharpoonup\mcal\times\ocal$ with domain $U$ such that $(\textbf{inp},F(\textbf{inp}))$ solves the fixed point equation \eqref{eq:lgpacfp}; we also say that $F$ is \emph{L-GPAC-generable}.
\end{enumerate}
\end{defi}

\begin{exa}\label{ex:lgpac}
To achieve a better understanding of Definition \ref{def:lgpacsemantics} we provide an example of an L-GPAC with one constant and four channels, as in Figure \ref{fig:lgpacex}. 

\begin{figure}[ht]
\centering
\begin{tikzpicture}
\draw (-.75, .75) -- ( .75, .75) -- ( .75,-.75) -- (-.75,-.75) -- (-.75, .75);
\draw (1.25, .75) -- (1.75, .75) -- (1.75, .25) -- (1.25, .25) -- (1.25, .75);
\draw (2.25, .75) -- (3.25, .75) -- (3.25,-.25) -- (2.25,-.25) -- (2.25, .75);
\draw[->] (-1.25, .5 ) -- (-.75, .5 );
\draw[->] (-1.25,-.5 ) -- (-.75,-.5 );
\draw[->] (  .75,0   ) -- (2.25,0   );
\draw[->] ( 1.5 ,0   ) -- (1.5 ,-1  ) -- (-1.75,-1  ) -- (-1.75, 0  ) -- (-.75, 0  );
\draw[->] ( 1.75, .5 ) -- (2.25, .5 );
\draw[->] ( 3.25, .25) -- (3.75, .25);

\node at ( 0  , 0 ) {\Large $\varint$};
\node at (1.5 , .5) {$x$};
\node at (2.75,.25) {\Large $\times$};
\node at (-1 , .75) {$g$};
\node at (-1 ,-.25) {$u_1$};
\node at (-1 , .25) {$u_2$};
\node at ( 2 , .25) {$u_2$};
\node at ( 2 , .75) {$u_3$};
\node at (3.5, .5 ) {$u_4$}; 
\end{tikzpicture}\caption{Example of an L-GPAC.\label{fig:lgpacex}}
\end{figure}
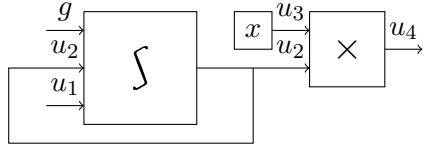

We assume that $\xcal=C(\rbb,\rbb)$ so that channels $u\in C^1(\tbb,\xcal)$ can be seen as functions of two variables, $u=u(t,x)$. Whenever $u$ is continuously differentiable (with respect to its first variable), we can write $u'(t,x)$ to denote its `time' derivative. Therefore, the integrator module over $C^1(\tbb,\xcal)$ has the following semantics: for an initial setting $g\in\xcal$ and stream inputs $u,v\in C^1(\tbb,\xcal)$, it outputs the stream $w\in C^1(\tbb,\xcal)$ given by
\[w(t,x)=g(x)+\int_0^tu(s,x)v'(s,x)ds.\]

In this example there is one constant input (labeled $g$, associated with the integrator module), one stream input, (labeled $u_1$), two mixed channels (labeled $u_2$ and $u_3$) and one output channel (labeled $u_4$). The induced operator simply formalizes the input/output relation between these channels,
\[\ical=\xcal\times C^1(\tbb,\xcal),\quad\mcal=C^1(\tbb,\xcal)^2,\quad\ocal=C^1(\tbb,\xcal);\]
\[\Phi:\ical\times\mcal\rightharpoonup\mcal\times\ocal\]
\[\Phi(g,u_1,u_2,u_3)=\left(g+\int u_2du_1,x,u_2u_3\right)=(\tilde{u_2},\tilde{u_3},u_4).\]
The corresponding fixed point equation is then given by the system
\begin{equation}u_2=g+\int_0^t u_2du_1,\quad u_3=x,\quad u_4=u_2 u_3;\end{equation}
which after some calculations can be seen to yield the solution
\begin{equation}\label{eq:lgpacexsol}u_2(t,x)=g(x)e^{u_1(t,x)-u_1(0,x)},\quad u_3(t,x)=x,\quad u_4(t,x)=xg(x)e^{u_1(t,x)-u_1(0,x)}.\end{equation}
Thus, this L-GPAC (actually a GPAC) is well-posed for any $u_1\in C^1(\tbb,\xcal)$ and $g\in\xcal$. It generates the function $F:\xcal\times C^1(\tbb,\xcal)\rightharpoonup C^1(\tbb,\xcal)^3$ given by $F(g,u_1)=(u_2,u_3,u_4)$, as in \eqref{eq:lgpacexsol}.
\end{exa}

\section{Preliminaries}

In this section we present some auxiliary results that are used in the later sections of the paper. We recall that a \emph{Fréchet space} is a vector space $\xcal$ equipped with a countable family of pseudonorms $\psnfamily$ such that $\xcal$ is complete with respect to $\psnfamily$ (that is, for all sequences $(x_m)$ such that $(x_m)$ is Cauchy with respect to each pseudonorm $\|\cdot\|_n$, there exists $x\in \xcal$ such that $(x_m)$ converges to $x$ with respect to each pseudonorm $\|\cdot\|_n$).\footnote{A detailed exposition of Fréchet spaces can be found in \cite[Chapter V]{reedsimons:80}.}

\begin{prop}[\textbf{Metric from family of pseudonorms}]\label{prop:indmetric} Let $\xcal$ be a Fréchet space and ${\psnfamily}$ a corresponding family of pseudonorms. Let $\gamma:\rbb_{\geq0}\rightarrow[0,1]$ be a continuous function which is also positive definite, nondecreasing and subadditive, that is,
\begin{itemize}
\item $\gamma(0)=0$ and for all $t>0$ we have $0<\gamma(t)\leq 1$;
\item for all $t_1,t_2\in\rbb_{\geq0}$ such that $t_1\leq t_2$ we have $\gamma(t_1)\leq\gamma(t_2)$;
\item for all $t_1,t_2\in\rbb_{\geq0}$ we have $\gamma(t_1+t_2)\leq\gamma(t_1)+\gamma(t_2)$.
\end{itemize}
Let $\{w_n\}_{n\in\nbb}$ be a summable family of positive weights, that is, $\ds\sum_{n=0}^\infty w_n<\infty$. Then we can define a metric on $\xcal$ by
\begin{equation} d(x,y)=\sum_{n=0}^\infty w_n\gamma(\|x-y\|_n).\end{equation}

Moreover, this metric induces the same topology over $\xcal$ and $\xcal$ is complete under it.
\end{prop}

\begin{prop}[\textbf{Bounds on the pseudonorms and bounds on the metric}]\label{prop:boundpnormmetric}
Let $\xcal$ be a Fréchet space with pseudonorms $\|\cdot\|_n$, $n\in\nbb^+$. Let $d$ be the metric on $\xcal$ given by
\begin{equation}\label{eq:dfromp}d(x,y)=\sum_{n=1}^\infty 2^{-n}\min(\|x-y\|_n,1).\end{equation}
\begin{enumerate}
\item Let $0<\epsilon<1$ and $M\in\nbb$. Then, for any $\delta\leq\epsilon 2^{-M}$ and $x,y\in \xcal$, one has
\begin{center}if $d(x,y)<\delta$, then $\|x-y\|_n<\epsilon$ for $n=1,\ldots,M$.\end{center}
\item Let $0<\epsilon<1$. Then for any $\delta\leq\epsilon/2$ and $M\in\nbb$ such that $2^{-M}\leq\epsilon/2$ and $x,y\in \xcal$, one has
\begin{center}if $\|x-y\|_n<\delta$ for $n=1,\ldots,M$, then $d(x,y)<\epsilon$.\end{center}
\end{enumerate}
\end{prop}
\noindent The proofs of the above two propositions are straightforward.

\begin{exa}\label{ex:inverter}
We exemplify the usefulness of the Shannon GPAC by building an \emph{inverter}, a partially defined function given by
\begin{equation}F:\rbb\times C^1([0,T],\rbb)\rightharpoonup C^1([0,T],\rbb);\quad F(k,b)(t)=\frac{k}{1+k(b(t)-b(0))}.\end{equation}

Let us show that $F$ is GPAC-generable by considering the GPAC in Figure \ref{fig:inverter}.

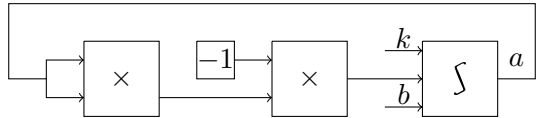
\begin{figure}[ht]
\centering
\begin{tikzpicture}
\draw (-.5, .5) -- ( .5, .5) -- ( .5,-.5) -- (-.5,-.5) -- (-.5, .5);
\draw (1  , .5) -- (1.5, .5) -- (1.5, 0 ) -- (1  , 0 ) -- (1  , .5);
\draw (2  , .5) -- (3  , .5) -- (3  ,-.5) -- (2  ,-.5) -- (2  , .5);
\draw (4  , .5) -- (5  , .5) -- (5  ,-.5) -- (4  ,-.5) -- (4  , .5);

\draw[->] ( .5,-.25) -- (2  ,-.25);
\draw[->] (1.5, .25) -- (2  , .25);
\draw[->] (3  , 0  ) -- (4  , 0  );
\draw[->] (3.5,-.375) -- (4  ,-.375);
\draw[->] (3.5, .375) -- (4  , .375);
\draw[->] (5  , 0  ) -- (5.5, 0  ) -- (5.5, 1  ) -- (-1.5,1) -- (-1.5,0) -- (-1,0) -- (-1,.25) -- (-.5,.25);
\draw[->] (-1 , 0  ) -- (-1 ,-.25) -- (-.5,-.25);

\node at ( 0   , 0  ) {$\times$};
\node at ( 1.25, .25) {$-1$};
\node at ( 2.5 , 0  ) {$\times$};
\node at ( 4.5 , 0  ) {$\varint$};
\node at ( 3.75, .55) {$k$};
\node at ( 3.75,-.2 ) {$b$};
\node at ( 5.25, .25) {$a$};
\end{tikzpicture}\caption{A GPAC generating the inverter functional.\label{fig:inverter}}
\end{figure}

This GPAC induces a system of four equations on six variables, which is reducible to a single equation on the channels labeled $k$, $a$ and $b$, given by
\[a'(t)=-a(t)^2b'(t),\quad a(0)=k;\]
after some calculations we find the unique solution to be
\[a(t)=\frac{k}{1+k(b(t)-b(0))}.\]

Therefore, $F$ is a (component of a) GPAC-generable partial function; its domain is given by
\[D(F)=\{(k,b)\in\rbb\times C^1(\tbb,\rbb):k=0\text{ or }b(t)\neq b(0)-1/k\text{ for all }t\in\tbb\}.\]

It is worth noticing that, when $k=1$ and $b=\mathbf{t}$ is linear time, the corresponding solution is $a(t)=\frac{1}{1+t}$, which provides an example for a GPAC-generable rational function.
\end{exa}

\section{Infinite speedup, infinite slowdown}\label{sec:infspeedslow}

The composition presented in Figure \ref{fig:limitcomposition} can be thought of as a \emph{time speedup} by $T$ (or \emph{slowdown}, if $T$ grows slower than the identity). The goal of this section is to observe that infinite speedups can also be expressed in our model. Thus, the choice of limit $t\rightarrow\infty$ is not the only possibility, as one may consider limits of the form $t\rightarrow T^-$ for any positive time $T\in\tbb$. In order to see this, we consider the following functions that continuously maps the interval $[0,1)$ to $[0,\infty)$ and vice versa.

\begin{prop}\label{prop:infinitespeedslow} The following functions are Shannon GPAC-generable:
\begin{enumerate}
\item $t\mapsto\frac{t}{1-t}$, with domain $[0,1)$ and range $[0,\infty)$;
\item $t\mapsto\frac{t}{1+t}$, with domain $[0,\infty)$ and range $[0,1)$. 
\end{enumerate}
\end{prop}

\begin{proof}
Recall the inverter functional $\Phi:(k,b)\mapsto a(t)=\frac{k}{1+k(b(t)-b(0))}$ constructed in Example \ref{ex:inverter}. The function $s_{\uparrow}(t)=\frac{1}{1-t}$ can be obtained as the output of $\Phi$ with $k=1$ and $b(t)=-t$. The function $s_{\downarrow}(t)=\frac{1}{1+t}$ can be obtained as the output of $\Phi$ with $k=1$ and $b(t)=t$. The desired functions can then be obtained by multiplying $s_{\uparrow}$ or $s_{\downarrow}$ with $t$.\qedhere

\begin{figure}[ht]\centering
\raisebox{-.5\height}{\includegraphics[width=0.3\textwidth]{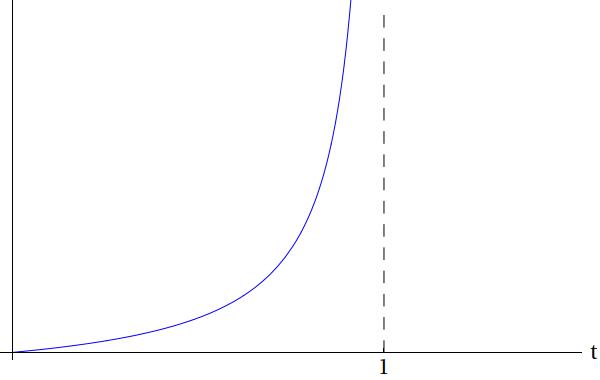}}
\raisebox{-.5\height}{\includegraphics[width=0.3\textwidth]{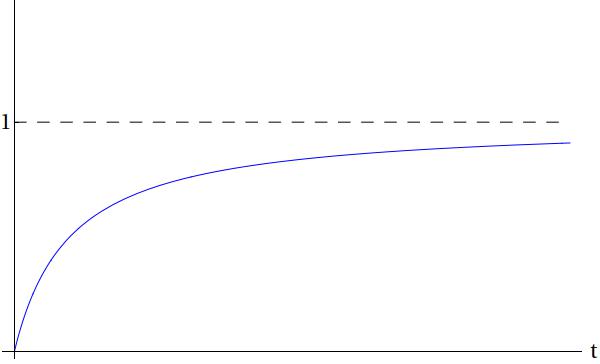}}
\caption[Infinite speedup and infinite slowdown.]{Plot of the functions $t\mapsto\frac{t}{1-t}$ (left) and $t\mapsto\frac{t}{1+t}$ (right).\label{fig:infinitespeedslow}}
\end{figure}
\end{proof}

Therefore, if we have a function $u(t)$ with a desired limit as $t\rightarrow\infty$, we can perform a composition of $u$ with the infinite speedup to obtain the desired limit as $t\rightarrow 1^-$. We must note that the reverse case is also possible; that is, if we consider a function that continuously maps the interval $[0,\infty)$ to $[0,1)$, then we can convert a limit as $t\rightarrow 1^-$ into a limit as $t\rightarrow+\infty$ via an infinite slowdown.

\section{Pseudonorm effectiveness}\label{sec:pseudoeffect}

Our construction of the limit module relies on the notion of effective limit, which is given by the metric associated to the underlying space $\xcal$. The advantage of this approach is that it requires only a minimal structure on $\xcal$ (complete metric space), and thus it can be applied quite generally. However, previous work \cite{pocaszucker:16,pocaszucker:18} provided evidence for the prevalence of Fréchet spaces in our research. Since the topology in these spaces is induced by a family of pseudonorms, we may desire to define a suitable notion of effective limits that takes this into consideration. Since a metric can be inferred from the pseudonorms (recall Proposition \ref{prop:indmetric}), we may expect some equivalence between both notions. In this section we formalize this argumentation.

\begin{defi}[\textbf{Moduli of convergence for pseudonorms}]\label{def:modconvpseudo}~
\begin{enumerate}
\item A \emph{discrete modulus of convergence for pseudonorms} is a function $N:\nbb\times\nbb\rightarrow\nbb$ such that for each $n\in\nbb$, $N(n,\cdot)$ is nondecreasing.
\item A \emph{continuous modulus of convergence for pseudonorms} is a function $T:\nbb\rightarrow C^1(\tbb,\rbb)$ such that for each $n\in\nbb$, $T(n)\in C^1(\tbb,\rbb)$ is nonnegative and nondecreasing.
\end{enumerate}
\end{defi}
\noindent Observe that for each $n\in\nbb$, the $n$-section of a (discrete or continuous) modulus of convergence for pseudonorms is itself a (discrete or continuous) modulus of convergence for the underlying space.

\begin{defi}[\textbf{Effective limits on Fréchet spaces}]\label{def:efflimfrechet}~
\begin{enumerate}
\item Let $N$ be a discrete modulus of convergence for pseudonorms and $\{g_n\}\in \xcal^\nbb$. Then $\{g_n\}$ is an \emph{$N$-Fréchet Cauchy sequence} (or an $N$-FC sequence) if
\[\text{for all }\nu\in\nbb,n\in\nbb\text{, for all }j,k\in\nbb\text{ with }j,k\geq N(n,\nu)\text{ one has }\|g_j-g_k\|_n<2^{-\nu}.\]
\item Let $T$ be a continuous modulus of convergence for pseudonorms and $u\in C^1(\tbb,\xcal)$. Then $u$ is a \emph{$T$-Fréchet Cauchy stream} (or a $T$-FC stream) if
\[\text{for all }\tau\in\tbb,n\in\nbb\text{, for all }s,t\in\tbb\text{ with }s,t\geq T(n,\tau)\text{ one has }\|u(s)-u(t)\|_n<2^{-\tau}.\]
\end{enumerate}
\end{defi}
\noindent For the following Lemma, we shall assume that the metric in $\xcal$ is induced by the pseudonorms as

\begin{equation}\label{eq:dfromp2}d(u,v)=\sum_{n\in\nbb}w_n\gamma(\|u-v\|_n),\quad\text{with }w_n=2^{-n}\text{ and }\gamma(t)=\min(t,1),\end{equation}
which satisfy the assumptions in Proposition \ref{prop:indmetric} (see also Proposition \ref{prop:boundpnormmetric}).

\begin{lem}[\textbf{Equivalence between effective limits}]\label{lem:convmodpseudoequiv}~
\begin{enumerate}
\item Let $N$ be a discrete modulus of convergence and $g\in \xcal^\nbb$ an $N$-convergent Cauchy sequence. Then $g$ is an $\tilde{N}$-FC sequence, where $\ds \tilde{N}(n,\nu)=N(n+\nu)$; moreover, if $N$ is computable, so is $\tilde{N}$.
\item Let $T$ be a continuous modulus of convergenge and $u\in C^1(\tbb,\xcal)$ a $T$-convergent Cauchy stream. Then $u$ is a $\tilde{T}$-FC stream, where $\tilde{T}(n,\tau)=T(n+\tau)$; moreover, if $T$ is GPAC-generable, so is $\tilde{T}(n)$ for each $n$.
\item Let $\tilde{N}$ be a discrete modulus of convergence for pseudonorms and $g\in \xcal^\nbb$ an $\tilde{N}$-FC sequence. Then $g$ is an $N$-convergent Cauchy sequence, where $\ds N(\nu)=\max_{n\leq\nu+1}\tilde{N}(n,\nu+1)$; moreover, if $\tilde{N}$ is computable, so is $N$.
\item Let $\tilde{T}$ be a continuous modulus of convergence for pseudonorms and $u\in C^1(\tbb,\xcal)$ a $\tilde{T}$-FC stream. Then $u$ is a $T$-convergent Cauchy stream, where \[T(\tau)=\max_{n\leq\tau+2}\tilde{T}(n,\tau+1).\]
\end{enumerate}
\end{lem}

\begin{proof}
To prove claim 1, we first observe that for each $n$, the function $\tilde{N}(n,\cdot):\nu\mapsto N(n+\nu)$ is nonnegative and nondecreasing (since $N$ is nonnegative and nondecreasing), so that $\tilde{N}$ is a discrete modulus of convergence for pseudonorms. It is also clear from inspection that if $N$ is computable, so is $\tilde{N}$.

Next, we take $\nu\in\nbb$, $n\in\nbb$ and $j,k\in\nbb$ with $j,k\geq\tilde{N}(n,\nu)$. By construction of $\tilde{N}$ this means that $j,k\geq N(n+\nu)$ and thus, since $g$ is an $N$-convergent Cauchy sequence, it follows that $d(g_j,g_k)<2^{-n-\nu}$. By applying Proposition \ref{prop:boundpnormmetric} we conclude that $\|g_j-g_k\|_n<2^{-\nu}$, so that $g$ is an $\tilde{N}$-FC sequence.

To prove claim 2, we first observe that for each $n$, the function $t\mapsto T(n+t)$ is nonnegative and nondecreasing (since $T$ is a nonnegative and nondecreasing), so that $\tilde{T}$ is a continuous modulus of convergence for pseudonorms. Moreover, each $t\mapsto T(n+t)$ is computable since it is the composition of $T$ with the function $t\mapsto t+n$, which can be obtained using one constant and one adder module. As a side remark, the procedure that maps $n$ into a GPAC $\gcal_n$ generating the corresponding $\tilde{T}(n)$ is also computable on $n$.

The remainder of the claim can be proved, \textit{mutatis mutandis}, as in claim 1.

To prove claim 3, we first see that the function $\ds\nu\mapsto\max_{n\leq\nu+1}\tilde{N}(n,\nu+1)$ is nonnegative and nondecreasing, since $\tilde{N}(n,\cdot)$ is nonnegative and nondecreasing for each $n$, so that $N$ is a discrete modulus of convergence. It is also clear that if $\tilde{N}$ is computable, so is $N$, since taking maxima is a computable operation in $\nbb$.

Next, we take $\nu\in\nbb$ and $j,k\in\nbb$ with $j,k\geq N(\nu)$. By construction of $N$ this means that $j,k\geq\tilde{N}(n,\nu+1)$ for all $n\leq\nu+1$ and thus, since $g$ is an $\tilde{N}$-FC sequence, it follows that $\|g_j-g_k\|_n<2^{-\nu-1}$ for all $n\leq\nu+1$. By applying Proposition \ref{prop:boundpnormmetric} we conclude that $d(g_j,g_k)<2^{-\nu}$, so that $g$ is an $N$-convergent Cauchy sequence.

To prove claim 4, we first see that the function $t\mapsto\max_{n\leq\tau+2}\tilde{T}(n,\tau+1)$ is nonnegative and nondecreasing, so that $T$ is a continuous module of convergence. The remainder of the claim can be proved, \textit{mutatis mutandis}, as in claim 3.
\end{proof}

\section{Computability of the Gamma function}\label{sec:gammafun}

Our motivation for considering limit operators is the computability of the gamma function,
\[\Gamma(x)=\int_0^\infty t^{x-1}e^{-t}dt,\]
which is not differentially algebraic (and thus, not Shannon GPAC-generable)\footnote{Proved in \cite{holder:86}, mentioned in \cite{shannon:41}.}. There are known differential equations in two variables related to the gamma function; for example (see \cite[p. 174]{olveretal:10}), if we define the incomplete gamma functions
\begin{equation}\label{eq:igamma}\gamma_{i1}(t,x)=\int_0^t s^{x-1}e^{-s}ds;\end{equation}
\begin{equation}\label{eq:Igamma}\gamma_{i2}(t,x)=\int_t^\infty s^{x-1}e^{-s}ds,\end{equation}
then both incomplete gamma functions satisfy the differential equation (for $w=w(t,x)$)
\begin{equation}\label{eq:gammad1}\frac{d^2w}{dt^2}+\left(1+\frac{1-x}{t}\right)\frac{dw}{dt}=0;\end{equation}
we shall now try to implement such relations on our analog networks.

\begin{figure}[ht]\centering
\includegraphics[width=0.6\textwidth]{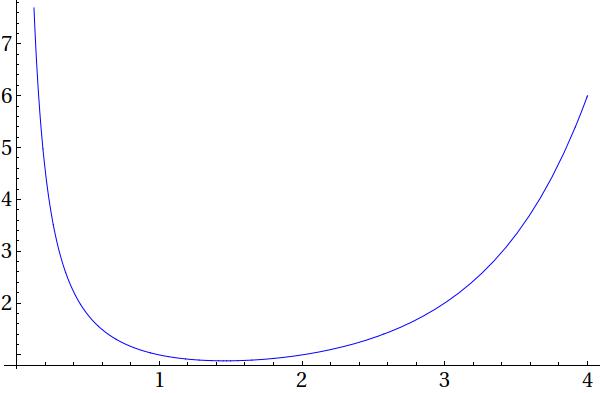}
\caption{Plot of the gamma function.\label{fig:gammafun}}
\end{figure}

Observe that GPAC models include a constant\footnote{That is, not dependent on the time variable $t$.} module for any function in $\xcal$, and in particular we could include a constant module for the gamma function itself! Of course, this is not an interesting way to obtain the gamma function.

The main idea is instead to obtain the gamma function as the limit of a function in two variables,
\[\Gamma(x)=\lim_{t\rightarrow\infty}\gamma(t,x),\]
for some function $\gamma\in C^1(\tbb,\xcal)$ which will be specified shortly. As remarked in Section \ref{sec:infspeedslow}, the choice of limit $t\rightarrow \infty$ is arbitrary, as we can take infinite speedups and consider, e.g., a limit $t\rightarrow 1^-$. Since $\Gamma(x)$ has a pole at $x=0$, we need to consider a space where functions are defined in a region ``away from'' $x=0$. For simplicity, we shall take $\xcal=C([1,+\infty),\rbb)$. Note that this is a Fréchet space with pseudonorms $\ds\|g\|_n=\sup_{1\leq x\leq n}|g(x)|$. We also observe that \eqref{eq:gammad1} is undetermined at $t=0$, and it would allow initial conditions $w|_{t=0}=\frac{dw}{dt}|_{t=0}=0$, for which $w\equiv 0$ is a different solution. Since well-posedness is desired, we must avoid starting at $t=0$; therefore, we consider integrals starting at $t=1$, writing
\[\Gamma(x)=\int_0^\infty t^{x-1}e^{-t}dt=\int_0^1 t^{x-1}e^{-t}dt+\int_1^\infty t^{x-1}e^{-t}dt.\]

The next step is to apply a change of variables in order to obtain integrals of the form $\int_0^\infty$; to be precise, we apply $t\mapsto s=\frac{1-t}{t}$ on the first integral and $t\mapsto s=t-1$ on the second integral, obtaining
\[\int_0^1 t^{x-1}e^{-t}dt=\int_0^\infty \left(\frac{1}{1+s}\right)^{x+1}e^{-1/(1+s)}ds=\lim_{t\rightarrow+\infty}\gamma_1(t,x);\]
\[\int_1^\infty t^{x-1}e^{-t}dt=\int_0^\infty (1+s)^{x-1}e^{-(1+s)}ds=\lim_{t\rightarrow+\infty}\gamma_2(t,x),\]
where
\[\gamma_1(t,x)=\int_0^t (1+s)^{-(x+1)}e^{-1/(1+s)}ds;\]
\[\gamma_2(t,x)=\int_0^t (1+s)^{x-1}e^{-(1+s)}ds.\]

We proceed to show that $\gamma_1$, $\gamma_2$ are L-GPAC-generable.

\textbf{Computation of $\gamma_1$}: by taking derivatives in time, we see that
\begin{align}
\frac{d\gamma_1}{dt}    &=      (1+t)^{-(x+1)}e^{-1/(1+t)};\nonumber \\
\frac{d^2\gamma_1}{dt^2}&=-(x+1)(1+t)^{-(x+2)}e^{-1/(1+t)}+(1+t)^{-(x+3)}e^{-1/(1+t)}=-\frac{x+xt+t}{(1+t)^2}\frac{d\gamma_1}{dt}\label{eq:gammad2};
\end{align}
moreover, we have initial conditions
\[\gamma_1(0,x)=0,\quad\frac{d\gamma_1}{dt}(0,x)=1/e.\]

We can look at the PDE \eqref{eq:gammad2} as an ODE in $t$ with a parameter $x$. It is then easy to check that it defines a well-posed problem since the multiplying factor $u_1(t,x)=-\frac{x+xt+t}{(1+t)^2}$ is defined for all $t\in\tbb$. As an intermediate step in generating $\gamma_1$ with an L-GPAC, via \eqref{eq:gammad2}, we generate the multiplying factor $u_1$, and to achieve this we consider the function $s_{\downarrow}(t)=\frac{1}{1+t}$, which is GPAC-generable by the proof of Proposition \ref{prop:infinitespeedslow}. We can thus construct $u_1=-(x+xt+t)s_{\downarrow}^2$ and obtain $\gamma_1$ with an L-GPAC as in Figure \ref{fig:gamma1}, which implements \eqref{eq:gammad2}.

\begin{figure}[ht]\centering
\raisebox{-.5\height}{\begin{tikzpicture}
\draw ( 0 ,-1  ) -- ( .5,-1  ) -- ( .5,-1.5) -- ( 0 ,-1.5) -- ( 0 ,-1  );
\draw ( 1 ,- .5) -- (1.5,- .5) -- (1.5,-1  ) -- ( 1 ,-1  ) -- ( 1 ,- .5);
\draw ( 2 ,  .5) -- (2.5,  .5) -- (2.5, 0  ) -- ( 2 , 0  ) -- ( 2 ,  .5);
\draw ( 2 ,-1  ) -- (2.5,-1  ) -- (2.5,-1.5) -- ( 2 ,-1.5) -- ( 2 ,-1  );
\draw ( 3 ,- .5) -- (3.5,- .5) -- (3.5,-1  ) -- ( 3 ,-1  ) -- ( 3 ,- .5);
\draw ( 3 ,-1.5) -- (3.5,-1.5) -- (3.5,-2  ) -- ( 3 ,-2  ) -- ( 3 ,-1.5);
\draw ( 4 ,  .5) -- (4.5,  .5) -- (4.5, 0  ) -- ( 4 , 0  ) -- ( 4 ,  .5);
\draw ( 4 ,-1  ) -- (4.5,-1  ) -- (4.5,-1.5) -- ( 4 ,-1.5) -- ( 4 ,-1  );
\draw ( 5 ,- .5) -- (5.5,- .5) -- (5.5,-1  ) -- ( 5 ,-1  ) -- ( 5 ,- .5);

\draw[->] ( 0 ,  .25) -- ( 2 ,  .25 );
\draw[->] ( .5,  .25) -- ( 1 ,- .625);
\draw[->] ( .5,  .25) -- ( 3 ,- .625);
\draw[->] ( .5,-1.25) -- ( 1 ,- .875);
\draw[->] ( .5,-1.25) -- ( 2 ,-1.25 );
\draw[->] (1.5,- .75) -- ( 2 ,-1.125);
\draw[->] (2.5,-1.25) -- ( 3 ,- .875);
\draw[->] (2.5,  .25) -- (3.5,  .25 ) -- ( 4 ,  .375);
\draw[->] (3.5,  .25) -- ( 4 ,  .125);
\draw[->] (3.5,- .75) -- ( 4 ,-1.125);
\draw[->] (3.5,-1.75) -- ( 4 ,-1.375);
\draw[->] (4.5,  .25) -- ( 5 ,- .625);
\draw[->] (4.5,-1.25) -- ( 5 ,- .875);
\draw[->] (5.5,- .75) -- ( 6 ,- .75 );

\node at ( .25 , .5 ) {$t$};
\node at ( .25,-1.25) {$x$};
\node at (1.25,- .75) {$\times$};
\node at (2.25,  .25) {$s_{\downarrow}$};
\node at (2.25,-1.25) {$+$};
\node at (3.25,- .75) {$+$};
\node at (3.25,-1.75) {$-1$};
\node at (4.25,  .25) {$\times$};
\node at (4.25,-1.25) {$\times$};
\node at (5.25,- .75) {$\times$};
\node at (5.75,- .5 ) {$u_1$};
\end{tikzpicture}}
\hspace{.5cm}
\raisebox{-.5\height}{\begin{tikzpicture}
\draw (-1.5, 0 ) -- (-1 ,0 ) -- (-1 ,-.5) -- (-1.5,-.5) -- (-1.5,0 );
\draw (- .5, .5) -- ( .5,.5) -- ( .5,-.5) -- (- .5,-.5) -- (- .5,.5);
\draw ( 1.5, .5) -- (2.5,.5) -- (2.5,-.5) -- ( 1.5,-.5) -- ( 1.5,.5);
\draw ( 3.5, .5) -- (4.5,.5) -- (4.5,-.5) -- ( 3.5,-.5) -- ( 3.5,.5);

\draw[->] (-2.5,-1  ) -- ( 3 ,-1  ) -- ( 3  ,-.25) -- (3.5,-.25);
\draw[->] (-2  ,-1  ) -- (-2 ,-.25) -- (-1.5,-.25);
\draw[->] ( 1  ,-1  ) -- ( 1 ,-.25) -- ( 1.5,-.25);
\draw[->] (-1  ,-.25) -- (-.5,-.25);
\draw[->] (  .5, .25) -- (1.5, .25);
\draw[->] ( 2.5, .25) -- (3.5, .25);
\draw[->] ( 3  , .25) -- ( 3 , 1  ) -- (-1  , 1  ) -- (-1 , .25) -- (-.5,.25);
\draw[->] ( 4.5, 0  ) -- ( 5 , 0  );

\node at (-1.25,-.25) {$u_1$};
\node at ( 0   , 0  ) {$\times$};
\node at ( 2   , 0  ) {$\varint$};
\node at ( 4   , 0  ) {$\varint$};
\node at (-2.25,-.75) {$t$};
\node at ( 1   , .5 ) {$\gamma_1''$};
\node at ( 2.75, .5 ) {$\gamma_1'$};
\node at ( 4.75, .25) {$\gamma_1$};
\end{tikzpicture}}
\caption[Construction of auxiliary functions $u_1(t)$ and $\gamma_1(t,x)$.]{Construction of $u_1(t)=-\frac{x+xt+t}{(1+t)^2}$ and $\gamma_1(t,x)$.\label{fig:gamma1}}
\end{figure}
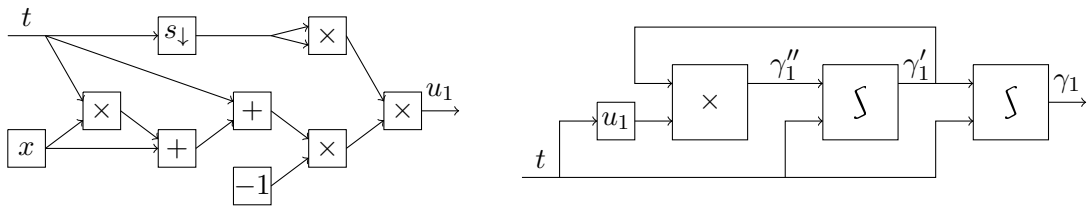

\textbf{Computation of $\gamma_2$}: by taking derivatives in time, we see that
\begin{align}
\frac{d\gamma_2}{dt}    &=     (1+t)^{x-1}e^{-(1+t)};\nonumber \\
\frac{d^2\gamma_2}{dt^2}&=(x-1)(1+t)^{x-2}e^{-(1+t)}-(1+t)^{x-1}e^{-(1+t)}=\frac{x-t-2}{1+t}\frac{d\gamma_2}{dt}\label{eq:gammad3};
\end{align}
moreover, we have initial conditions
\[\gamma_2(0,x)=0,\quad\frac{d\gamma_2}{dt}(0,x)=1/e.\]

As with $\gamma_1$, we can look at the PDE \eqref{eq:gammad3} as an ODE in $t$ with a parameter $x$. It is then easy to check that it defines a well-posed problem, since the multiplying factor $u_2(t,x)=\frac{x-t-2}{1+t}$ is defined for all $t\in\tbb$. As an intermediate step in generating $\gamma_2$ with an L-GPAC, via \eqref{eq:gammad3}, we generate the multiplying factor $u_2$, and to achieve this we again consider the function $s_{\downarrow}(t)=\frac{1}{1+t}$ from the proof of Proposition \ref{prop:infinitespeedslow}. We can thus construct $u_2=(x-t-2)s_{\downarrow}$ and obtain $\gamma_2$ with an L-GPAC as in Figure \ref{fig:gamma2}, which implements \eqref{eq:gammad3}.

\begin{figure}[ht]\centering
\raisebox{-.5\height}{\begin{tikzpicture}
\draw ( 0 ,-1  ) -- ( .5,-1  ) -- ( .5,-1.5) -- ( 0 ,-1.5) -- ( 0 ,-1  );
\draw ( 0 ,-2  ) -- ( .5,-2  ) -- ( .5,-2.5) -- ( 0 ,-2.5) -- ( 0 ,-2  );
\draw ( 0 ,-3  ) -- ( .5,-3  ) -- ( .5,-3.5) -- ( 0 ,-3.5) -- ( 0 ,-3  );
\draw ( 1 ,- .5) -- (1.5,- .5) -- (1.5,-1  ) -- ( 1 ,-1  ) -- ( 1 ,- .5);
\draw ( 1 ,-2.5) -- (1.5,-2.5) -- (1.5,-3  ) -- ( 1 ,-3  ) -- ( 1 ,-2.5);
\draw ( 2 , 0  ) -- (2.5, 0  ) -- (2.5,- .5) -- ( 2 ,- .5) -- ( 2 , 0  );
\draw ( 2 ,-1.5) -- (2.5,-1.5) -- (2.5,-2  ) -- ( 2 ,-2  ) -- ( 2 ,-1.5);
\draw ( 3,-.75) -- (3.5,-.75) -- (3.5,-1.25) -- ( 3,-1.25) -- ( 3 ,-.75);

\draw[->] ( 0 ,- .25) -- ( 2 ,-  .25);
\draw[->] ( .5,- .25) -- ( 1 ,- .625);
\draw[->] ( .5,-1.25) -- ( 1 ,- .875);
\draw[->] ( .5,-2.25) -- ( 1 ,-2.625);
\draw[->] ( .5,-3.25) -- ( 1 ,-2.875);
\draw[->] (1.5,- .75) -- ( 2 ,-1.625);
\draw[->] (1.5,-2.75) -- ( 2 ,-1.875);
\draw[->] (2.5,- .25) -- ( 3 ,- .875);
\draw[->] (2.5,-1.75) -- ( 3 ,-1.125);
\draw[->] (3.5,-1   ) -- ( 4 ,-1    );

\node at ( .25, 0   ) {$t$};
\node at ( .25,-1.25) {$-1$};
\node at ( .25,-2.25) {$-2$};
\node at ( .25,-3.25) {$x$};
\node at (1.25,- .75) {$\times$};
\node at (1.25,-2.75) {$+$};
\node at (2.25,- .25) {$s_{\downarrow}$};
\node at (2.25,-1.75) {$+$};
\node at (3.25,-1   ) {$\times$};
\node at (3.75,- .75) {$u_2$};
\end{tikzpicture}}
\hspace{.5cm}
\raisebox{-.5\height}{\begin{tikzpicture}
\draw (-1.5, 0 ) -- (-1 ,0 ) -- (-1 ,-.5) -- (-1.5,-.5) -- (-1.5,0 );
\draw (- .5, .5) -- ( .5,.5) -- ( .5,-.5) -- (- .5,-.5) -- (- .5,.5);
\draw ( 1.5, .5) -- (2.5,.5) -- (2.5,-.5) -- ( 1.5,-.5) -- ( 1.5,.5);
\draw ( 3.5, .5) -- (4.5,.5) -- (4.5,-.5) -- ( 3.5,-.5) -- ( 3.5,.5);

\draw[->] (-2.5,-1  ) -- ( 3 ,-1  ) -- ( 3  ,-.25) -- (3.5,-.25);
\draw[->] (-2  ,-1  ) -- (-2 ,-.25) -- (-1.5,-.25);
\draw[->] ( 1  ,-1  ) -- ( 1 ,-.25) -- ( 1.5,-.25);
\draw[->] (-1  ,-.25) -- (-.5,-.25);
\draw[->] (  .5, .25) -- (1.5, .25);
\draw[->] ( 2.5, .25) -- (3.5, .25);
\draw[->] ( 3  , .25) -- ( 3 , 1  ) -- (-1  , 1  ) -- (-1 , .25) -- (-.5,.25);
\draw[->] ( 4.5, 0  ) -- ( 5 , 0  );

\node at (-1.25,-.25) {$u_2$};
\node at ( 0   , 0  ) {$\times$};
\node at ( 2   , 0  ) {$\varint$};
\node at ( 4   , 0  ) {$\varint$};
\node at (-2.25,-.75) {$t$};
\node at ( 1   , .5 ) {$\gamma_2''$};
\node at ( 2.75, .5 ) {$\gamma_2'$};
\node at ( 4.75, .25) {$\gamma_2$};
\end{tikzpicture}}
\caption[Construction of auxiliary functions $u_2(t)$ and $\gamma_2(t,x)$.]{Construction of $u_2(t)=\frac{x-t-2}{1+t}$ and $\gamma_2(t,x)$.\label{fig:gamma2}}
\end{figure}
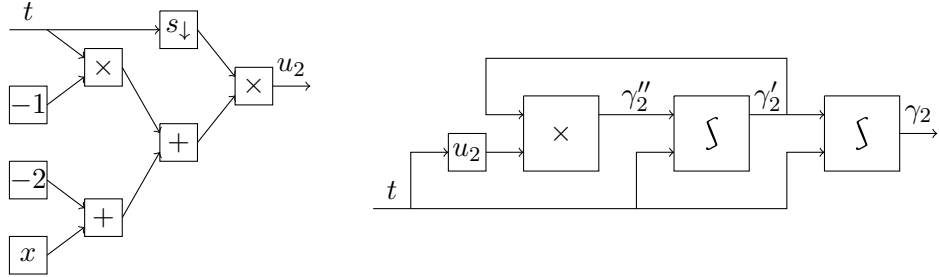

\textbf{Construction of $\Gamma$}: We finally obtain $\Gamma(x)$ as the limit
\[\Gamma(x)=\lim_{t\rightarrow\infty}\gamma_1(t,x)+\gamma_2(t,x),\]

which can be obtained using a continuous limit module. However, we must still find an effective modulus of convergence for our approximation. This will be done with two technical lemmas.

\begin{lem}\label{lem:gamma1conv}
Let $T\in\tbb$. For any $x\in[1,+\infty)$ and any $t_1,t_2\geq T$ one has
\[|\gamma_1(t_1,x)-\gamma_1(t_2,x)|\leq 1/T.\]
\end{lem}

\begin{proof}
Under the assumptions of the lemma, we have
\begin{align*}
|\gamma_1(t_1,x)-\gamma_1(t_2,x)|&=\left|\int_{t_2}^{t_1}(1+s)^{-(x+1)}e^{-1/(1+s)}ds\right|<\int_T^\infty(1+s)^{-(x+1)}e^{-1/(1+s)}ds\\
&<\int_T^\infty(1+s)^{-2}ds=\frac{1}{1+T}<\frac{1}{T}.\tag*{\qedhere}
\end{align*}
\end{proof}

\begin{lem}\label{lem:gamma2conv}
Let $T\in\tbb$, and $k\in\nbb$. For any $x\in[1,k+1]$ and any $t_1,t_2\geq T$ one has
\[|\gamma_2(t_1,x)-\gamma_2(t_2,x)|\leq (k+1)!(T+1)^ke^{-(T+1)}.\]
\end{lem}

\begin{proof}
Under the assumptions of the lemma, we have
\begin{align*}
|\gamma_2(t_1,x)-\gamma_2(t_2,x)|&=\left|\int_{t_2}^{t_1}(1+s)^{x-1}e^{-(1+s)}ds\right|<\int_T^\infty(1+s)^{x-1}e^{-(1+s)}ds\\
&\leq\int_T^\infty(1+s)^ke^{-(1+s)}ds=\int_{T+1}^\infty s^ke^{-s}ds\leq(k+1)!(T+1)^ke^{-(T+1)},
\end{align*}
where the last inequality can be proved by induction on $k\in\nbb$ (using integration by parts).
\end{proof}

Therefore, the limits in $\gamma_1$, $\gamma_2$ become effective for suitable moduli of convergence. We can merge these two results and prove effectiveness of our construction, as in the next result. We continue to use the metric given by \eqref{eq:dfromp2}. Recall that $\xcal=C([1,\infty),\rbb)$ is a Fréchet space with pseudonorms $\ds\|g\|_n=\sup_{1\leq x\leq n}|g(x)|$.

\begin{lem}\label{lem:compgammafun}
Let $\gamma=\gamma_1+\gamma_2$, where $\gamma_1$, $\gamma_2$ are defined as in \eqref{eq:gammad2}, \eqref{eq:gammad3}. Then $\ds\lim_{t\rightarrow\infty}\gamma(t)=\Gamma$; moreover, $\gamma$ is a $T$-convergent Cauchy stream for $T(\tau)=C 2^\tau$ with a suitably large constant $C$.
\end{lem}

\begin{proof}
Only the effectiveness of the limit remains to be proven. Let $\tau\in\tbb$, $T=T(\tau)=C2^\tau$ and take $t_1,t_2\in\tbb$ with $t_1,t_2\geq T$. We can write $d(\gamma(t_1),\gamma(t_2))\leq d(\gamma_1(t_1),\gamma_1(t_2))+d(\gamma_2(t_1),\gamma_2(t_2))$ and thus we can treat $\gamma_1$ and $\gamma_2$ separately.

To deal with $\gamma_1$, we use Lemma \ref{lem:gamma1conv} to conclude that, for any $n\in\nbb^+$, we have
\[\|\gamma_1(t_1)-\gamma_1(t_2)\|_n\leq \frac{1}{T}\leq\frac{1}{C}2^{-\tau};\]
thus, we obtain the bound
\begin{align*}
d(\gamma_1(t_1),\gamma_1(t_2))&=\sum_{n=1}^\infty 2^{-n}\min\{1,\|\gamma_1(t_1)-\gamma_1(t_2)\|_n\}\\
&\leq\sum_{n=1}^\infty 2^{-n}\|\gamma_1(t_1)-\gamma_1(t_2)\|_n\leq\sum_{n=1}^\infty 2^{-n}\frac{1}{C}2^{-\tau}=\frac{1}{C}2^{-\tau},
\end{align*}
which is smaller than $2^{-\tau-1}$ for a suitably large $C$ (namely, for $C>2$).

To deal with $\gamma_2$, we use Lemma \ref{lem:gamma2conv} to conclude that, for any $n\in\nbb^+$, we have
\[\|\gamma_2(t_1)-\gamma_2(t_2)\|_n\leq n!(T+1)^{n-1}e^{-(T+1)};\]
next, we shall take $N=\lceil\tau\rceil+2$, so that $\tau+2\leq N < \tau+3$. By splitting the sum, we obtain the bound
\begin{subequations}
\begin{align}
d(\gamma_2(t_1),\gamma_2(t_2))
&=\sum_{n=1}^\infty 2^{-n}\min\{1,\|\gamma_1(t_1)-\gamma_1(t_2)\|_n\}\\
&\leq\sum_{n=1}^N 2^{-n}\|\gamma_1(t_1)-\gamma_1(t_2)\|_n+\sum_{n=N+1}^\infty 2^{-n}\\
&\leq\sum_{n=1}^N 2^{-n}n!(T+1)^{n-1}e^{-(T+1)}+2^{-N}\\
&\leq N!(T+1)^{N-1}e^{-(T+1)}\sum_{n=1}^N 2^{-n}+2^{-\tau-2}\\
&<eN^{N+1/2}e^{-N}(T+1)^{N-1}e^{-(T+1)}+2^{-\tau-2}\label{eq:boundgamma2}\\
&=\exp\{(N+1/2)\log(N)-N+(N-1)\log(T+1)-T\}+2^{-\tau-2}\\
&<\exp\{(\tau+7/2)\log(\tau+3)-\tau-2+(\tau+2)\log(C2^\tau+1)-C2^\tau\}+2^{-\tau-2},\label{eq:boundgamma3}
\end{align}
\end{subequations}
where \eqref{eq:boundgamma2} is justified by Stirling's approximation \cite[Chapter 8]{rudin:76},

\begin{equation}\label{eq:stirling}\sqrt{2\pi}k^{k+\frac{1}{2}}e^{-k}\leq k!\leq ek^{k+\frac{1}{2}}e^{-k}.\end{equation}

For the last step, we wish to bound the first term as
\[\exp\{(\tau+7/2)\log(\tau+3)-\tau-2+(\tau+2)\log(C2^\tau+1)-C2^\tau\}\leq\exp\{-\log(2)(\tau+2)\},\]
which means that we need to find $C$ such that, for all $\tau\in\tbb$,
\[(\tau+7/2)\log(\tau+3)-\tau-2+(\tau+2)\log(C2^\tau+1)-C2^\tau\leq-\log(2)(\tau+2).\]
But this certainly holds for a suitably large $C$ that does not depend on $\tau$, because the term $C2^\tau$ largely dominates all other terms (numerically, we have found that $C>2.85216$ suffices). Thus the bound in \eqref{eq:boundgamma3} can be further taken to be smaller than $\exp\{-\log(2)(\tau+2)\}+2^{-\tau-2}=2^{-\tau-1}$.

Combining the two bounds, we conclude that $d(\gamma(t_1),\gamma(t_2))<2^{-\tau}$ and therefore $\gamma$ is a $T$-convergent Cauchy stream.
\end{proof}

\begin{thm}\label{thm:compgammafun}
The gamma function is L-GPAC-generable.
\end{thm}

\begin{proof}
By Lemma \ref{lem:compgammafun}, the gamma function $\Gamma$ can be seen as the $T$-convergent limit of some function $\gamma$. Moreover, by the preceding discussion, both $\gamma$ and $T$ can be seen to be L-GPAC-generable; in particular, $\gamma$ is the sum of two L-GPAC-generable functions. Thus we can devise an L-GPAC that generates $\Gamma$, as in Figure \ref{fig:gammalimit}.
\end{proof}
\begin{figure}[ht]\centering
\begin{tikzpicture}
\draw (-3  ,  .5 ) -- (-2  ,  .5 ) -- (-2  ,- .5 ) -- (-3  ,- .5 ) -- (-3  ,  .5 );
\draw (- .5, 1.25) -- (  .5, 1.25) -- (  .5,  .25) -- (- .5,  .25) -- (- .5, 1.25);
\draw (- .5,- .25) -- (  .5,- .25) -- (  .5,-1.25) -- (- .5,-1.25) -- (- .5,- .25);
\draw ( 1.5,  .5 ) -- ( 2.5,  .5 ) -- ( 2.5,- .5 ) -- ( 1.5,- .5 ) -- ( 1.5,  .5 );
\draw ( 3.5,  .5 ) -- ( 4.5,  .5 ) -- ( 4.5,- .5 ) -- ( 3.5,- .5 ) -- ( 3.5,  .5 );
\draw[->] (-4  , 0  ) -- (-3  , 0  );
\draw[->] (-2  , 0  ) -- (-1  , 0  ) -- (-1  , .75) -- (-.5, .75);
\draw[->] (-1  , 0  ) -- (-1  ,-.75) -- (-.5 ,-.75);
\draw[->] (  .5, .75) -- ( 1  , .75) -- ( 1  , .25) -- (1.5, .25);
\draw[->] (  .5,-.75) -- ( 1  ,-.75) -- ( 1  ,-.25) -- (1.5,-.25);
\draw[->] ( 2.5, 0  ) -- ( 3.5, 0  );
\draw[->] ( 4.5, 0  ) -- ( 5.5, 0  );

\node at (-2.5, 0  ) {$T$};
\node at ( 0  , .75) {$\gamma_1$};
\node at ( 0  ,-.75) {$\gamma_2$};
\node at ( 2  , 0  ) {$+$};
\node at ( 4  , 0  ) {$\lcal_c$};
\node at (-3.5, .25) {$t$};
\node at (-1.5, .25) {$T(t)$};
\node at ( 5  , .25) {$\Gamma$};
\end{tikzpicture}
\caption[Construction of the gamma function.]{Construction of the gamma function; $T$ indicates an exponential speedup $T(\tau)=C2^\tau$.\label{fig:gammalimit}}
\end{figure}
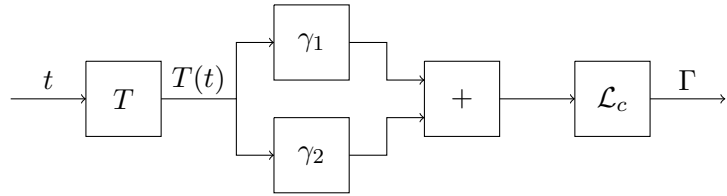

\section{Computability of the Riemann zeta function}\label{sec:compzetafun}

Our next case study concerns the computation of the Riemann zeta function, which for complex numbers with real part greater than 1 is given by

\begin{equation}\label{eq:rzeta}\zeta(z)=\sum_{n=1}^\infty\frac{1}{n^z}.\end{equation}

This function has a pole at $z=1$ and thus we should consider a space of functions defined in a region ``away from'' $z=1$. In particular, we take $\xcal=C([2,\infty),\rbb)$; in other words, we shall be interested in computing $\zeta(x)$ for real values of $x$ larger or equal than $2$. Note that $\xcal$ is a Fréchet space with pseudonorms $\ds \|g\|_n=\sup_{1\leq x\leq n}|g(x)|$.

We need a representation of the Riemann zeta function that is amenable to our framework of analog networks. Fortunately, there are known integral representations that we can use, such as

\begin{equation}\label{eq:rzetagamma}\zeta(x)=\frac{1}{\Gamma(x)}\int_0^\infty\frac{t^{x-1}}{e^t-1}dt,\end{equation}
or the Abel-Plana formula \cite{abel:65,plana:20}

\begin{equation}\label{eq:rzetaap}\zeta(x)=\frac{2^x}{x-1}-2^x\int_0^\infty\frac{\sin(x\arctan t)}{(1+t^2)^{x/2}(e^{\pi t+1})}dt.\end{equation}

The latter formula will allow us to express the zeta function as the limit of a function in two variables,
\[\zeta(x)=\lim_{t\rightarrow\infty}\zeta_1(t,x),\]
for a function $\zeta_1$ which computes the bounded integral
\begin{equation}\label{eq:zeta1}\zeta_1(t,x)=\frac{2^x}{x-1}-2^x\int_0^t\frac{\sin(x\arctan s)}{(1+s^2)^{x/2}(e^{\pi s+1})}ds.\end{equation}

For such a function, we have $\zeta_1(0,x)=\frac{2^x}{x-1}$ and $\frac{d\zeta_1}{dt}=-2^x\zeta_2$, where
\begin{equation}\label{eq:zeta2}\zeta_2(t,x)=\frac{\sin(x\arctan t)}{(1+t^2)^{x/2}(e^{\pi t+1})}.\end{equation}

\begin{lem}\label{lem:zeta2}
The function $\zeta_2$ defined in \eqref{eq:zeta2} is L-GPAC-generable.
\end{lem}
\begin{proof}
This requires several steps, so we just provide a sketch of the construction:

\begin{enumerate}
\item the function $t\mapsto\frac{1}{1+t^2}$ is GPAC-generable; it can be given as the output of the inverter (from Example \ref{ex:inverter}) with inputs $k=1$ and $b(t)=t^2$;
\item the function $t\mapsto\arctan t$ is GPAC-generable; observe that $(\arctan t)'=\frac{1}{1+t^2}$ and use step 1;
\item the function $(t,x)\mapsto\sin(x\arctan t)$ is L-GPAC-generable; compose $(t,x)\mapsto x\arctan t$ (from step 2) with $t\mapsto\sin(t)$;
\item the function $(t,x)\mapsto(1+t^2)^{-x/2}$ is L-GPAC-generable; if $u(t,x)=(1+t^2)^{-x/2}$ then $\frac{du}{dt}=-\frac{xt}{1+t^2}u$, with $u(0,x)=1$; use step 1 (see Figure \ref{fig:zetastepfour} for more details);
\item the function $t\mapsto e^{-\pi t-1}$ is GPAC-generable; compose $t\mapsto -\pi t-1$ with $t\mapsto e^t$;
\item the function $\zeta_2$ is L-GPAC-generable; write $\zeta_2(t,x)=\sin(x\arctan t)(1+t^2)^{-x/2}e^{-\pi t-1}$ and use steps 3, 4, 5.
\end{enumerate}
\end{proof}

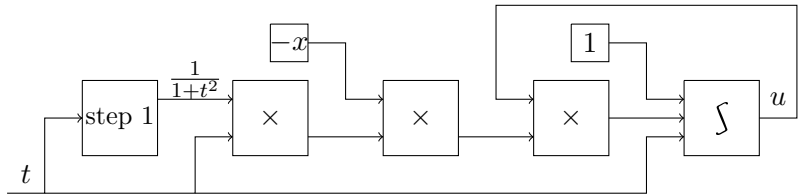
\begin{figure}[ht]\centering
\begin{tikzpicture}
\draw (0,.5) -- (1,.5) -- (1,-.5) -- (0,-.5) -- (0,.5);
\draw (2,.5) -- (3,.5) -- (3,-.5) -- (2,-.5) -- (2,.5);
\draw (4,.5) -- (5,.5) -- (5,-.5) -- (4,-.5) -- (4,.5);
\draw (6,.5) -- (7,.5) -- (7,-.5) -- (6,-.5) -- (6,.5);
\draw (8,.5) -- (9,.5) -- (9,-.5) -- (8,-.5) -- (8,.5);
\draw (2.5,1.25) -- (3,1.25) -- (3,.75) -- (2.5,.75) -- (2.5,1.25);
\draw (6.5,1.25) -- (7,1.25) -- (7,.75) -- (6.5,.75) -- (6.5,1.25);
\draw[->] (-1 ,-1  ) -- (-.5,-1  ) -- (-.5, 0  ) -- (0,0);
\draw[->] ( 1 , .25) -- (2  , .25);
\draw[->] (-.5,-1  ) -- (1.5,-1  ) -- (1.5,-.25) -- (2  ,-.25);
\draw[->] ( 3 , 1  ) -- (3.5, 1  ) -- (3.5, .25) -- (4  , .25);
\draw[->] ( 3 ,-.25) -- (4  ,-.25);
\draw[->] ( 5 ,-.25) -- (6  ,-.25);
\draw[->] ( 7 , 0  ) -- (8  , 0  );
\draw[->] ( 7 , 1  ) -- (7.5, 1  ) -- (7.5, .25) -- (8  , .25);
\draw[->] (1.5,-1  ) -- (7.5,-1  ) -- (7.5,-.25) -- (8  ,-.25);
\draw[->] ( 9 , 0  ) -- (9.5, 0  ) -- (9.5,1.5 ) -- (5.5,1.5 ) -- (5.5,.25) -- (6,.25);

\node at (  .5 ,0) {\small step 1};
\node at ( 2.5 ,0) {$\times$};
\node at ( 4.5 ,0) {$\times$};
\node at ( 6.5 ,0) {$\times$};
\node at ( 8.5 ,0) {$\varint$};
\node at ( 2.75,1) {$-x$};
\node at ( 6.75,1) {$1$};

\node at (-.75,-.75) {$t$};
\node at ( 1.5, .5 ) {$\frac{1}{1+t^2}$};
\node at (9.25, .25) {$u$};
\end{tikzpicture}
\caption[Intermediate step in the Riemann zeta function construction.]{Construction of the function $u(t,x)=(1+t^2)^{-x/2}$ appearing in step 4 of the proof of Lemma \ref{lem:zeta2}.\label{fig:zetastepfour}}
\end{figure}

\begin{thm}\label{thm:compzetafun}
The Riemann zeta function is L-GPAC-generable.
\end{thm}

\begin{proof}
We can obtain $\zeta_1$ (from \eqref{eq:zeta1}) by feeding $\zeta_2$ (which is L-GPAC-generable by Lemma \ref{lem:zeta2}) into an integrator module and using constants $\frac{2^x}{x-1}$, $2^x$. We can obtain the Riemann zeta function by feeding $\zeta_1$ into an effective limit module. Thus we can devise an L-GPAC that generates $\zeta$, as in Figure \ref{fig:rzetalimit}.

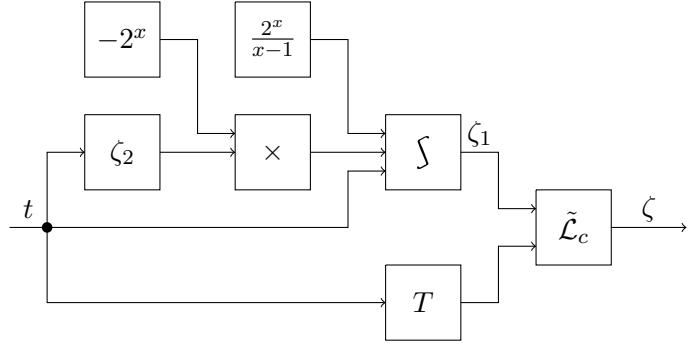
\begin{figure}[ht]\centering
\begin{tikzpicture}
\draw ( 0, 3  ) -- ( 1, 3  ) -- ( 1, 2  ) -- ( 0, 2  ) -- ( 0, 3  );
\draw ( 0, 1.5) -- ( 1, 1.5) -- ( 1,  .5) -- ( 0,  .5) -- ( 0, 1.5);
\draw ( 2, 3  ) -- ( 3, 3  ) -- ( 3, 2  ) -- ( 2, 2  ) -- ( 2, 3  );
\draw ( 2, 1.5) -- ( 3, 1.5) -- ( 3,  .5) -- ( 2,  .5) -- ( 2, 1.5);
\draw ( 4, 1.5) -- ( 5, 1.5) -- ( 5,  .5) -- ( 4,  .5) -- ( 4, 1.5);
\draw ( 4,- .5) -- ( 5,- .5) -- ( 5,-1.5) -- ( 4,-1.5) -- ( 4,- .5);
\draw ( 6,  .5) -- ( 7,  .5) -- ( 7,- .5) -- ( 6,- .5) -- ( 6,  .5);
\draw[->] (-1  , 0  ) -- ( 3.5, 0  ) -- ( 3.5, .75) -- ( 4  , .75);
\draw[->] (- .5, 0  ) -- (- .5, 1  ) -- ( 0  , 1  );
\draw[->] ( 1  , 1  ) -- ( 2  , 1  );
\draw[->] ( 1  , 2.5) -- ( 1.5, 2.5) -- ( 1.5,1.25) -- ( 2  ,1.25);
\draw[->] ( 3  , 1  ) -- ( 4  , 1  );
\draw[->] ( 3  , 2.5) -- ( 3.5, 2.5) -- ( 3.5,1.25) -- ( 4  ,1.25);
\draw[->] (- .5, 0  ) -- (- .5,-1  ) -- ( 4  ,-1  );
\draw[->] ( 5  , 1  ) -- ( 5.5, 1  ) -- ( 5.5, .25) -- ( 6  , .25);
\draw[->] ( 5  ,-1  ) -- ( 5.5,-1  ) -- ( 5.5,-.25) -- ( 6  ,-.25);
\draw[->] ( 7  , 0  ) -- ( 8  , 0  );

\node at (  .5,2.5 ) {$-2^x$};
\node at (  .5,1   ) {$\zeta_2$};
\node at ( 2.5,2.5 ) {$\frac{2^x}{x-1}$};
\node at ( 2.5,1   ) {$\times$};
\node at ( 4.5,1   ) {$\varint$};
\node at ( 4.5,-1  ) {$T$};
\node at ( 6.5, 0  ) {$\tilde{\lcal}_c$};

\node at (- .75, .25) {$t$};
\node at ( 5.25,1.25) {$\zeta_1$};
\node at ( 7.5 , .25) {$\zeta$};

\fill (-.5,0) circle (2pt);
\end{tikzpicture}
\caption[Construction of the Riemann zeta function.]{Construction of the Riemann zeta function; $T$ denotes a suitable continuous modulus of convergence.\label{fig:rzetalimit}}
\end{figure}

The only thing left is to prove the effectiveness of the convergence. In order to do that we shall prove that a linear modulus of convergence $T(\tau)=C\tau$, for a suitable large constant $C$, is sufficient. The following calculations are similar to those done for Lemmas \ref{lem:gamma1conv}, \ref{lem:gamma2conv} and \ref{lem:compgammafun}. To start, we recall that $\xcal=C([2,\infty),\rbb)$ is a Fréchet space with pseudonorms $\ds \|g\|_n=\sup_{2\leq x\leq n}|g(x)|$. Let $T\in\tbb$, $k\in\nbb$ with $k\geq 2$, $x\in[2,k]$ and $t_1,t_2\in\tbb$ with $t_1,t_2\geq T$; then we have the bound
\begin{align*}
|\zeta_1(t_1,x)-\zeta_1(t_2,x)|
&=\left|2^x\int_{t_2}^{t_1}\frac{\sin(x \arctan t)}{(1+t^2)^{x/2}e^{\pi t+1}}dt\right|\\
&<2^x\int_T^\infty\left|\frac{\sin(x \arctan t)}{(1+t^2)^{x/2}e^{\pi t+1}}\right|dt\\
&\leq 2^k\int_T^\infty\frac{1}{e^{\pi t+1}}dt=\frac{2^k}{e\pi}e^{-\pi T}.
\end{align*}
Thus, for any $k\geq 2$ and any $t_1,t_2\geq T(\tau)$ we have

\begin{equation}\label{eq:boundzeta1}\|\zeta_1(t_1)-\zeta_1(t_2)\|_k\leq\frac{2^k}{e\pi}e^{-\pi T}=\frac{2^k}{e\pi}e^{-\pi C\tau}.\end{equation}

Next, let us take $N=\lceil\tau\rceil+1$, so that $\tau+1\leq N<\tau+2$. By splitting the sum, we obtain
\begin{subequations}
\begin{align}
d(\zeta_1(t_1),\zeta_1(t_2))
&=\sum_{n=2}^\infty 2^{-n}\min(\|\zeta_1(t_1)-\zeta_1(t_2)\|_n,1)\\
&\leq\sum_{n=2}^N 2^{-n}\|\zeta_1(t_1)-\zeta_1(t_2)\|_n+\sum_{n=N+1}^\infty 2^{-n}\\
&\leq\sum_{n=2}^N 2^{-n}\frac{2^n}{e\pi}e^{-\pi C\tau}+2^{-N}\label{eq:boundzeta2}\\
&\leq\frac{N-1}{e\pi}e^{-\pi C\tau}+2^{-\tau-1}\\
&<\frac{\tau+1}{e\pi}e^{-\pi C\tau}+2^{-\tau-1}\\
&=\exp\{-\pi C\tau+\log(\tau+1)-\log(e\pi)\}+2^{-\tau-1},\label{eq:boundzeta3}
\end{align}
\end{subequations}
where \eqref{eq:boundzeta2} is justified by \eqref{eq:boundzeta1}. Finally, the expression in \eqref{eq:boundzeta3} can be further taken to be smaller than $\exp\{-\log(2)(\tau+1)\}+2^{-\tau-1}=2^{-\tau}$ for a suitably large $C$ that does not depend on $\tau$, because the term $\pi C\tau$ dominates all other terms (numerically, we have found that $C>0.25079$ suffices). Thus
\[d(\zeta_1(t_1),\zeta_1(t_2))<2^{-\tau},\]
so that $\zeta_1$ is a $T$-convergent Cauchy stream. Incidentally, since $C$ can be chosen to be equal to 1, the stream $\zeta_1$ is $\id$-convergent (meaning that the one-input continuous limit module can be applied directly).
\end{proof}

\section{Conclusion and further work}
In this paper we introduced limit modules to the Shannon GPAC computational model, arriving at a generalization which we called L-GPAC. The main motivation was to prove that some non-differentially algebraic functions such as the gamma function can be generated in this framework. In some sense, that result was obtained before (see \cite{graca:04}) by changing the notion of GPAC-generability to allow for approximability of functions.

The idea of \emph{approximability} is a cornerstone in many models of computability on continuous spaces, especially those that use classically computable functions (i.e. computable functions on the naturals) as a starting point. This is a consequence of the fact that many continuous spaces are typically represented using a dense countable subset and codes of convergent sequences. Then, to say that a function is computable is to assert that its values can be obtained up to a prescribed precision in an effective way.

As we have seen, the limit module is an operation of type $C(\tbb,\xcal)\rightarrow \xcal$ whose output is of `one arity less' than the input. This can be seen as the reverse of the integrator module, whose initial constant $g\in \xcal$ is of `one arity less' than the output in $C(\tbb,\xcal)$. This suggests one possible way of defining a hierarchy of `computable functions' on $\xcal$ based on the $\xcal$-GPAC model presented in \cite{pocaszucker:18}, which we sketch as follows (see also Figure \ref{fig:lgpacrec}):

\begin{enumerate}
\item Assume $\xcal=C(\rbb,\rbb)$ and take the subset $\xcal_0\subseteq \xcal$ of Shannon GPAC-generable functions (ignore momentarily the fact that $C(\rbb,\rbb)$ is not the same as the class $C^1(\tbb,\rbb)$ appearing in the Shannon GPAC); in this way, $\xcal_0$ corresponds to the class of differentially algebraic real functions as proven by Shannon and others.
\item Using $\xcal_0$ as a class of `valid initial constants' for integration in an $\xcal$-GPAC, define a class of $\xcal$-GPAC-generable functions in $C([0,\infty),\xcal)$.
\item Using continuous limit modules (that is, the L-GPAC framework), define a subclass $\xcal_1\subseteq \xcal$ of valid limits of the $\xcal$-GPAC-generable functions from the previous step; observe that this new class contains the gamma and Riemann zeta function, so $\xcal_1$ is strictly larger than $\xcal_0$.
\item The procedure can be iterated to get a hierarchy $\xcal_0\subseteq \xcal_1\subseteq \xcal_2\subseteq\ldots\subseteq \xcal$ of `computable functions' on $\xcal$.
\end{enumerate}

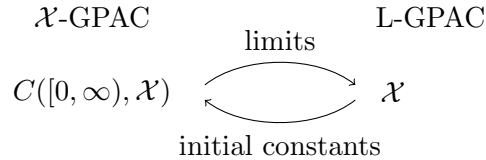
\begin{figure}[ht]
\centering
\begin{tikzpicture}
\draw[->] (-1, .1) .. controls (-.5, .5) and ( .5, .5) .. ( 1, .1);
\draw[->] ( 1,-.1) .. controls ( .5,-.5) and (-.5,-.5) .. (-1,-.1);
\node at (-2.5,1) {$\xcal$-GPAC}; 
\node at (-2.5,0) {$C([0,\infty),\xcal)$};
\node at ( 2  ,1) {L-GPAC};
\node at ( 1.5,0) {$\xcal$};
\node at (0, .7) {limits};
\node at (0,-.7) {initial constants};
\end{tikzpicture}
\caption{Recursive definition of a hierarchy of GPAC-generable functions.\label{fig:lgpacrec}}
\end{figure}

\noindent We leave as an open problem the task of defining this hierarchy precisely and studying its properties. We conjecture that this hierarchy actually collapses at level $n=1$, which would mean that using the limit operation just once is enough to capture the full class of L-GPAC-generable functions. Our intuition comes from the discrete case, where we know that, given an effective sequence of fast effective sequences of approximations, a diagonal argument can produce a new fast effective sequence. This argument also appears in the framework of $\alpha$-tracking computability, where it is used to show that $C_{\bar{\alpha}}(C_{\bar{\alpha}}(X))=C_{\bar{\alpha}}(X)$ \cite[Remark 8.1.1]{tuckerzucker:04} and it should be applicable to the case of continuous limits as well. Nevertheless, we hope that the union $\ds \bigcup_{n\in\nbb} \xcal_n$ is different from $\xcal$, in order to have a non-trivial model of computation.

Another direction for further research would consist in comparing our model of computation (the L-GPAC) with other models of computability in continuous spaces. For example, we could look at the notion of tracking computability presented in \cite{tuckerzucker:04} and find out suitable conditions under which the functions generated by a GPAC are tracking computable (and vice-versa). There is a large volume of research and literature dedicated to the task of defining an analog counterpart to the Church-Turing thesis, and this could be considered as an important step towards that goal.

We also remark that similar results have been achieved; for example, the paper \cite{bournezetal:07} already shows an equivalence between a GPAC model (which includes approximability) and computable analysis. Thus, with some care, their techniques may be adaptable to our framework.

We hope that in tackling these problems new insights can be acquired about the power of analog networks, and in particular the GPAC, as a model of analog computability.

\textbf{Acknowledgements.} The research of Diogo Poças was funded by Fundação para a Ciência e Tecnologia, Portugal (Doctoral Grant) and by the Alexander von Humboldt Foundation. The research of Jeffery Zucker was funded by the Natural Sciences and Engineering Research Council of Canada.

\bibliographystyle{alpha}
\bibliography{references.bib}{}

\begin{thebibliography}{BCGH07}

\bibitem[Abe65]{abel:65}
Niels~Henrik Abel.
\newblock Solution de quelques problèmes à l’aide d’intégrales
  définies.
\newblock In Ludwig Sylow and Sophus Lie, editors, {\em Oeuvres complètes
  d’Abel}, volume~I, pages 11--27. Johnson, New York, 1965.
\newblock Reprint of the Nouvelle éd., Christiania, 1881.

\bibitem[BCGH06]{olivieretal:06}
Olivier Bournez, Manuel~L. Campagnolo, Daniel~S. Gra{\c c}a, and Emmanuel
  Hainry.
\newblock The general purpose analog computer and computable analysis are two
  equivalent paradigms of analog computation.
\newblock In {Jin-yi} Cai, S.~Barry Cooper, and Angsheng Li, editors, {\em
  Theory and Applications of Models of Computation, Third International
  Conference, {TAMC} 2006, Beijing, China, May 15-20, 2006, Proceedings},
  volume 3959 of {\em Lecture Notes in Computer Science}, pages 631--643.
  Springer, 2006.

\bibitem[BCGH07]{bournezetal:07}
Olivier Bournez, Manuel~L. Campagnolo, Daniel~S. Graça, and Emmanuel Hainry.
\newblock Polynomial differential equations compute all real computable
  functions on computable compact intervals.
\newblock {\em Journal of Complexity}, 23(3):317--335, 2007.

\bibitem[Bus31]{bush:31}
Vannevar Bush.
\newblock The differential analyzer. a new machine for solving differential
  equations.
\newblock {\em Journal of the {F}ranklin {I}nstitute}, 212(4):447--488, 1931.

\bibitem[GC03]{gracacosta:03}
Daniel {Gra\c ca} and {Jos\'e}~{F\'elix} Costa.
\newblock Analog computers and recursive functions over the reals.
\newblock {\em Journal of {C}omplexity}, 19(5):644--664, 2003.

\bibitem[{Gra}04]{graca:04}
Daniel {Gra\c ca}.
\newblock Some recent developments on shannon's general purpose analog
  computer.
\newblock {\em Mathematical Logic Quarterly}, 50(4--5):473--485, 2004.

\bibitem[Grz55]{grzegorczyk:55}
A.~Grzegorczyk.
\newblock Computable functions.
\newblock {\em Fundamenta Mathematicae}, 42:168--202, 1955.

\bibitem[Grz57]{grzegorczyk:57}
A.~Grzegorczyk.
\newblock On the defintions of computable real continuous functions.
\newblock {\em Fundamenta Mathematicae}, 44:61--71, 1957.

\bibitem[Har50]{hartree:50}
Douglas~R. Hartree.
\newblock {\em Calculating instruments and machines}.
\newblock Cambridge University Press, 1950.

\bibitem[H{\"o}l86]{holder:86}
Otto H{\"o}lder.
\newblock Ueber die eigenschaft der gammafunction keiner algebraischen
  differentialgleichung zu gen{\"u}gen.
\newblock {\em Mathematische Annalen}, 28(1):1--13, 1886.

\bibitem[Hol96]{holst:96}
Per~A. Holst.
\newblock Svein {R}osseland and the {O}slo {A}nalyzer.
\newblock {\em IEEE Annals of the History of Computing}, 18(4):16--26, 1996.

\bibitem[Joh96]{johansson:96}
Magnus Johansson.
\newblock Early analog computers in {S}weden - with examples from {C}halmers
  {U}niversity of {T}echnology and the {S}wedish aerospace industry.
\newblock {\em IEEE Annals of the History of Computing}, 18(4):27--33, 1996.

\bibitem[JZ13]{jameszucker:13}
Nick~D. James and Jeffery~I. Zucker.
\newblock A class of contracting stream operators.
\newblock {\em The Computer Journal}, 56:15--33, 2013.

\bibitem[Ko91]{ko:91}
Ker-I Ko.
\newblock {\em Complexity Theory of Real Functions}.
\newblock Birk{\"a}user, 1991.

\bibitem[Lac55a]{lacombe:55a}
D.~Lacombe.
\newblock Extension de la notion de fonction récursive aux fonctions d'une ou
  plusieurs variables réelles i.
\newblock {\em Comptes Rendus des Séances d l'Académie des Sciences, Paris},
  240:2478--2480, 1955.

\bibitem[Lac55b]{lacombe:55b}
D.~Lacombe.
\newblock Extension de la notion de fonction récursive aux fonctions d'une ou
  plusieurs variables réelles ii.
\newblock {\em Comptes Rendus des Séances d l'Académie des Sciences, Paris},
  241:13--14, 1955.

\bibitem[Lac55c]{lacombe:55c}
D.~Lacombe.
\newblock Extension de la notion de fonction récursive aux fonctions d'une ou
  plusieurs variables réelles iii.
\newblock {\em Comptes Rendus des Séances d l'Académie des Sciences, Paris},
  241:151--153, 1955.

\bibitem[LR87]{lipshitzrubel:87}
Leonard Lipshitz and Lee Rubel.
\newblock A differentially algebraic replacement theorem.
\newblock {\em Proceedings of the American Mathematical Society},
  99(2):367--372, 1987.

\bibitem[OLBC10]{olveretal:10}
Frank W.~J. Olver, Daniel~W. Lozier, Ronald~F. Boisvert, and Charles~W. Clark.
\newblock {\em NIST Handbook of Mathematical Functions}.
\newblock Cambridge University Press, 2010.

\bibitem[PE74]{pourel:74}
Marian Pour-El.
\newblock Abstract computability and its relations to the general purpose
  analog computer.
\newblock {\em Transactions of the American Mathematical Society}, 199:1--28,
  1974.

\bibitem[PER79]{pourelrichards:79}
Marian Pour-El and Ian Richards.
\newblock A computable ordinary differential equation which possesses no
  computable solution.
\newblock {\em Annals of Mathematical Logic}, 17:61--90, 1979.

\bibitem[Pla20]{plana:20}
Giovanni Antonio~Amedeo Plana.
\newblock Sur une nouvelle expression analytique des nombres bernoulliens,
  propre à exprimer en termes finis la formule générale pour la sommation
  des suites.
\newblock {\em Mem. Accad. Sci. Torino}, 1(25):403--418, 1820.

\bibitem[PZ16]{pocaszucker:16}
Diogo Poças and Jeffery Zucker.
\newblock Fixed point techniques in analog systems.
\newblock {\em Mathematical and Computational Approaches in Advancing Modern
  Science and Engineering}, pages 701--711, 2016.

\bibitem[PZ18]{pocaszucker:18}
Diogo Poças and Jeffery Zucker.
\newblock Analog networks in function data spaces.
\newblock {\em Computability}, 7(4):301--322, 2018.

\bibitem[RS80]{reedsimons:80}
Michael Reed and Barry Simon.
\newblock {\em Methods of modern mathematical physics: Functional analysis}.
\newblock Academic Press, Inc, 1980.

\bibitem[Rud76]{rudin:76}
Walter Rudin.
\newblock {\em Principles of Mathematical Analysis}.
\newblock International Series in Pure and Applied Mathematics. McGraw-Hill,
  3rd edition, 1976.

\bibitem[Sha41]{shannon:41}
Claude Shannon.
\newblock Mathematical theory of the differential analyser.
\newblock {\em Journal {M}athematical {P}hysics}, 20:337--354, 1941.

\bibitem[SHT99]{hansentucker:99}
Viggo Stoltenberg-Hansen and John Tucker.
\newblock Concrete models of computation for topological algebras.
\newblock {\em Theoretical Computer Science}, 219:347--378, 1999.

\bibitem[Sma93]{small:93}
James~S. Small.
\newblock General-purpose electronic analog computing: 1945-1965.
\newblock {\em IEEE Annals of the History of Computing}, 15(2):8--18, 1993.

\bibitem[TT80]{thompsontait:80}
William Thompson and Peter~G. Tait.
\newblock {\em Treatise on Natural Philosophy}, volume~1.
\newblock Cambridge University Press, 2nd edition, 1880.
\newblock Part I.

\bibitem[TZ04]{tuckerzucker:04}
John~V. Tucker and Jeffery~I. Zucker.
\newblock Abstract versus concrete computation on metric partial algebras.
\newblock {\em ACM Transactions on Computational Logic}, 5:611--668, 2004.

\bibitem[TZ07]{tuckerzucker:07}
John~V. Tucker and Jeffery~I. Zucker.
\newblock Computability of analog networks.
\newblock {\em Theoretical Computer Science}, 371:115--146, 2007.

\bibitem[TZ11]{tuckerzucker:11}
John~V. Tucker and Jeffery~I. Zucker.
\newblock Continuity of operators on continuous and discrete time streams.
\newblock {\em Theoretical Computer Science}, 412:3378--3403, 2011.

\bibitem[TZ14]{tuckerzucker:14}
John~V. Tucker and Jeffery~I. Zucker.
\newblock Computability of operators on continuous and discrete time streams.
\newblock {\em Computability}, 3:9--44, 2014.

\bibitem[Wei00]{weihrauch:00}
Klaus Weihrauch.
\newblock {\em Computable Analysis --- An Introduction}.
\newblock Texts in Theoretical Computer Science. Springer-Verlag Berlin
  Heidelberg, 2000.

\end{thebibliography}

\end{document}